\documentclass[10pt,twocolumn,journal]{IEEEtran}

\usepackage{latexsym}
\usepackage{amsfonts}
\usepackage{amsbsy}
\usepackage{graphicx}
\usepackage{subfigure}
\usepackage{float}
\usepackage{algorithm}
\usepackage{algorithmic}
\usepackage{flushend}
\usepackage{mathrsfs}
\usepackage{amsfonts}
\usepackage{array}
\usepackage{amssymb,amsmath}
\usepackage{cases}
\usepackage{longtable}
\usepackage{rotating}
\usepackage{multirow}
\usepackage{epstopdf}
\usepackage{flushend}
\usepackage{cite}
\usepackage{float}
\usepackage[marginal]{footmisc}
\usepackage{stfloats}
\usepackage[colorlinks,
            linkcolor=black,
            anchorcolor=black,
            urlcolor=blue,
            citecolor=black]{hyperref}
\usepackage[hyphenbreaks]{breakurl}

\newtheorem{theorem}{Theorem}

\newtheorem{lemma}{Lemma}
\newtheorem{proof}{Proof}

\input epsf
\begin{document}

\title{Reflection Resource Management for Intelligent Reflecting Surface Aided Wireless Networks }
\author{Yulan Gao, \IEEEmembership{}
Chao Yong, \IEEEmembership{}
Zehui Xiong, \IEEEmembership{Student Member, IEEE,}
Jun Zhao, \IEEEmembership{Member, IEEE,}
Yue Xiao, \IEEEmembership{Member, IEEE,}
Dusit Niyato, \IEEEmembership{Fellow, IEEE}
\thanks{

Y. Gao is with the National Key Laboratory of Science
and Technology on Communications, University of Electronic
Science and Technology of China, Chengdu 611731, China, and also with
the School of Computer Science and Engineering, Nanyang Technological University,
Singapore.

C. Yong is with the National Key Laboratory of Science
and Technology on Communications, University of Electronic
Science and Technology of China, Chengdu 611731, China.

Z. Xiong, J. Zhao, and D. Niyato are with the School
of Computer Science and Engineering, Nanyang Technological University,
Singapore.

Y. Xiao is with the National Key Laboratory of Science
and Technology on Communications, University of Electronic
Science and Technology of China, Chengdu 611731, China,
and also  with the National Mobile Communications
Research Laboratory, Southeast University, Nanjing 210096, China.

}
}

\markboth{~}
\maketitle
\vspace{-3.3em}
\maketitle
\begin{abstract}
In this paper, the adoption of an intelligent reflecting surface (IRS) for multiple user pairs in two-hop networks is investigated.
Different from the existing studies on IRS that merely focused on tuning the reflection coefficients of all elements, we consider the \textit{true} reflection resource  management (RRM) which builds on the premise of the introduced modular IRS structure,  which is composed of multiple independent and controllable modules.
In particular, the \textit{true} RRM can be realized via the best triggered module subset identification.
To provide fairness among users, our goal, therefore, is to maximize the minimum signal-to-interference-plus-noise ratio (SINR) at destination terminals (DTs) via joint triggered module subset identification, transmit power allocation, and the corresponding passive beamforming, subject to per source terminal (ST) power budgets and module size constraint.
Whereas this problem is NP-hard due to the module size constraint,  to deal with it, we transform it into the group sparse constraint by introducing the mixed $\ell_{1,F}\text{-norm},$ which yields a suitable semidefinite relaxation.
In order to address this approximated problem, a two-block alternating direction method of multipliers (ADMM) algorithm is proposed based on its separable structure.
Numerical simulations are used to validate the analysis and assess the performance of the proposed algorithm as a function of the system parameters.
In addition,  energy efficiency (EE) performance comparisons demonstrate the necessity and meaningfulness of the introduced modular IRS structure.
Specifically, for a given network setting, there is an optimal value of the number of triggered modules, when the EE is considered.

\end{abstract}
\IEEEpeerreviewmaketitle

\section{Introduction}

\subsection{Background and Motivation}
\IEEEPARstart{I}{t is} evident that the fast prolific spread of Internet-enabled mobile devices will bring to a 1000-fold increment of network capacity by 2020 \cite{whitepaper}, which can not be supported by the forth-generation (4G) mobile networks, i.e., long-term evolution (LTE) and LTE-advanced (LTE-A) technologies.
Therefore, a lot of attention from research community was mainly focused on the design of fifth generation (5G) wireless technologies, e.g., heterogeneous networks (HetNets), peer-to-peer (P2P) communications, massive multiple-input multiple-output (mMIMO), and mmWave communication, which should address high quality of service (QoS), coverage, seamless connectivity with a high user speed, and limited power consumption \cite{Grassi2017Uplink, Chen2014The, Boccardi2014Five}.
However, due to the highly demanding of forthcoming and future wireless networks (5G and beyond), a serious issue in the wireless industry today is to meet the soaring demand at the cost of resulting power consumption \cite{Grassi2017Uplink, NGMN2015}.
For instance, for mMIMO, adopting a higher amount of base station antennas to serve multiple users concurrently not only entails the increased radio frequency chains and maintenance cost, but also significantly decreases the overall performance level.
Therefore, addressing this issue means introducing innovation technologies in future/\mbox{beyond-5G} wireless networks, which are spectral-energy efficient and cost-effective \cite{Zenzo2019Wireless, Wu2019Towards}.
As a result, intelligent reflecting surface (IRS) was treated as a promising innovation  technology for future/\mbox{beyond-5G} wireless networks supporting reconfigurable wireless environment via exploiting large software-controlled reflecting elements \cite{Tan2018Enabling, Liu2019Intelligent, Li2017Electromagnetic, Di2019Hybrid, Gao2020Reconfigurable}.
The IRS provides a new degree of freedom to further enhance the wireless link performance via proactively steering the incident radio-frequency wave towards destination terminals (DTs) as its important feature, which is a promising solution to build a programmable wireless environment for 5G and beyond systems.

The IRS-aided communications refer to the scenario that a large number of software-controlled reflecting elements  with adjustable phase shifts for reflecting the incident signal.
As such,  the phase shifts of all reflecting elements can be tuned adaptively according to the state of networks, e.g., the channel conditions and the incident angle of the signal by the source terminal (ST).
Notably, different from the conventional half and full-duplex modes, in IRS-aided communications, the propagation environment can be improved without incurring additional noise at the reflecting elements.
Currently, considerable research attention was paid for IRS-aided communications \cite{Wu2019Towards, Wu2018Intelligent,Guo2019Weighted, Wu2019Beamforming, Bjornson2019Intelligent, Han2018Large,Wu2019Intelligentjournal,Huang2019Reconfigurable, Zenzo2019Reconfigurable, huang2020holographic}.
Among the early contributions in this area, \cite{Wu2019Towards, Zenzo2019Reconfigurable, huang2020holographic} summarized the main  applications and competitive advantages in IRS-aided systems.
For the IRS-aided point-to-point multiple-input-single-output (MISO) wireless system with single user, \cite{Wu2018Intelligent} investigated the total received signal power maximization problem by jointly optimizing the transmit beamforming and the passive beamforming.
In the spirit of these works, a vast corpus of literature focused on optimizing active-passive beamforming for unilateral spectral efficiency (SE) maximization subject to power constraint.
For instance, \cite{Guo2019Weighted} proposed a fractional programming based alternating optimization approach to maximize the weighted SE in IRS-aided MISO downlink communication systems.
In particular, three assumptions for the feasible set of reflection coefficient (RC) were consider at IRS, including the ideal RC constrained by peak-power, continuous phase shifter, and discrete phase shifter.
Meantime, in MISO wireless systems, the problem of minimizing the total transmit power at the access point was considered to
energy-efficient active-passive beamforming \cite{Wu2019Beamforming,Wu2019Intelligentjournal }.
\cite{Wu2019Beamforming} formulated and solved the total transmit power minimization problem by  active-passive beamforming design, subject to the signal-to-interference-plus-noise ratio (SINR) constraints, where each reflecting element is a continuous phase shifter.
Along this direction, considering the discrete reflect phase shifts at the IRS,  the same optimization problem was  studied in \cite{Wu2019Intelligentjournal}.
{In addition, to reduce implementation complexity of large dimensions optimization problems, the deep reinforcement learning technique was incorporated into optimal designs for IRS-aided MIMO systems \cite{huang2020reconfigurable}.
}
Notably, the aforementioned studies for IRS-aided communications were based on the assumption that the IRS power consumption is ignored.
In contrast, in \cite{Huang2019Reconfigurable},  an energy efficiency (EE) maximization problem was investigated by developing a realistic IRS power consumption model, where IRS power consumption relies on the type and the resolution of meta-element.

The common assumption in the existing studies for IRS-aided communications is that all the reflecting elements reflect the incident signal together, i.e., adjusting RC of each meta-element simultaneously.
However, along with the use of a large number of high-resolution reflecting elements, especially with continuous phase shifters, triggering all the reflecting elements every time may result in significant  IRS power consumption \cite{Mehanna2013Joint} and high implementation complexity.
Moreover, the hardware support for the IRS implementation is the use of a large number of tunable metasurfaces.
Specifically, the tunability feature can be realized by introducing mixed-signal integrated circuits (ICs) or diodes/varactors, which can  vary both the resistance and reactance, offering complete
local control over the complex surface impedance \cite{Liu2019Intelligent,Tan2018Enabling, Li2017Electromagnetic, Zenzo2019Smart}.
According to the IRS power consumption model presented in \cite{Huang2019Reconfigurable} and the hardware support, triggering the entire IRS all the time not only incurs increased power consumption, but also entails the increased latency of adjusting phase-shift  and accelerates equipment depreciation.
To this end, it is necessary and valuable to achieve the \textit{true} RRM in IRS-aided systems.

\subsection{Novelty and Contribution}
In this paper, we consider the two-hop P2P network in which multiple single-antenna STs reach the corresponding single-antenna DTs through an IRS that forwards a suitably phase-shifted version of the transmitted signal.
The goal is to maximize the minimum SINR at DTs via joint design of triggered module subset identification, transmit power allocation, and the corresponding passive beamformer.
Specifically, the novelty and contributions of this paper are summarized in the following aspects.
\begin{enumerate}
\item{ \textit{Modular IRS Structure:} { For the first time, we introduce a modular IRS structure, which composes of  multiple independent and controllable modules.
    At the modular IRS, each module is attached with a smart controller, where all the controllers are physically connected via dedicated fiber links, thus, they can exchange information with each other in a point-to-point fashion.
    IRS can be programmable and controlled by the controller, and hence,  from an operational standpoint, independent module triggering can be implemented easily.
    }
}
 \item{\textit{Triggered Module Subset Identification:}
{The \textit{true} RRM can be realized via the triggered module subset identification, which is based on the introduced modular IRS structure where the reflecting elements are triggered independently by the controller to which they belong.
 Specifically, we formulate the problem of maximizing the minimum SINR subject to the maximum transit power at each ST and the module size constraint.
 Unfortunately, due to the module size constraint, obtaining a globally optimal solution  requires an exhaustive combinatorial search over all possible cases, where the NP-hard max-min SINR problem must be solved for each of these cases.
To this end, we resort to a low-complexity and efficient approximate solution.
In particular, a convex relaxation to the hard module size constraint can be derived by replacing the $\ell_0\text{-norm}$ by the convex norm $\ell_{1, F}\text{-norm}$.
  Based on this insight, a tractable convex problem can be formulated from the perspective of group sparse optimization.
 To select the best triggered module subset, a two-block alternating direction method of multipliers (ADMM) algorithm is proposed based on the separable structure of the dual of the approximate convex programming.
  Notably, the global convergence of the proposed two-block ADMM algorithm can be established.
  Subsequently, transmit power allocation and the corresponding passive beamformer design for the max-min SINR problem without the module size constraint are studied while simultaneously meeting STs' power budget.
  In addition, the entire algorithm (outlined in Alg. 2) exhibits low complexity, since eqs. (\ref{s:24-1}), (\ref{s:29-1}), (\ref{s:33-1}), (\ref{s:34-1}), and (\ref{s:34-2}) are derived analytically in closed-form and can be computed directly.
 }
 }

\item{\textit{Performance Comparison:}
{Numerical simulations are used to validate the analysis and assess the
performance of the two-block ADMM algorithm as a function of the system parameters.  Detailed results are provided first when the number of modules is small, which allows comparing the ADMM algorithm with the method of exhaustive search (MES) solution.
The results show that the two-block ADMM algorithm for the approximate convex problem converges to a near optimal solution.
To gain insight into the introduced modular IRS structure, energy efficiency (EE) performance comparison is given subsequently, where the EE is defined as the ratio of achievable sum rate and the overall power consumption.
The further simulation results show that there exists an optimal value of triggered modules for a given network setting, which implies that the introduction of the modular IRS structure is meaningful.
}
}
\end{enumerate}

\subsection{Paper Outline and Notation}
The reminder of this paper is organized as follows. The system model, triggered module subset identification, the convex relaxation for the original problem, and the approximate optimization problem formulation are presented in Section II.
Section III presents the proposed two-block ADMM algorithm.
Section IV reports numerical results that are used to assess the performance of the proposed algorithm and further demonstrate the meaningfulness of the modular IRS structure.
Conclusions are presented in Section V.

Matrices and vectors are denoted by bold letters. ${\mathbf I}_N$, ${\mathbf 0}_N,$ and ${\mathbf e}_n$ are the $N\times N$ identity matrix, the $N\times 1$ all-zero column vector, and the $N\times 1$  elementary vector with a one at the $n \text{th}$ position, respectively.
${\mathbf A}^{T}, {\mathbf A}^{\dag},$ ${\mathbf A}^{-1}$, and $||{\mathbf A}||_F$ denote transpose, Hermitian, inverse, and Frobenius norm of matrix ${\mathbf A},$ respectively.
{Notation $\text{bldg}\{{\mathbf A}^1, \ldots, {\mathbf A}^M\}$ denotes a block diagonal matrix with ${\mathbf A}^m$ being the $m\text{-th}$ diagonal block,  and $[{\mathbf A}]_{i,j}$ denotes the $(i,j)\text{th}$ entry of ${\mathbf A}.$
$[{\mathbf A}, {\mathbf a}]$ and $[{\mathbf a}, a]$ denote the splicing of matrix and vector as well as vector and scalar, respectively.}
$\text{Re}(\cdot),$ $\text{Im}(\cdot),$ and  $|\cdot|$  denote the real part, imaginary part, and modulus of the enclosed vector, respectively.


\begin{figure}[!t]
\centering
\begin{minipage}[t]{1\linewidth}
\centering
\includegraphics[width=1\linewidth]{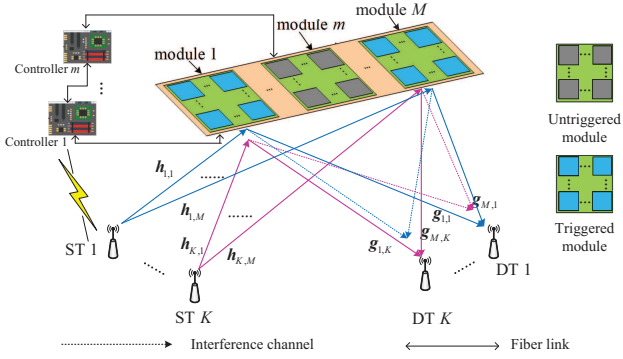}
\end{minipage}
\caption{An IRS-aided P2P network with $K$ ST-DT pairs and an IRS, which has $M$  modules.     }
\label{fig:1}
\end{figure}
\section{System Model and Problem Formulation }

\subsection{System Model}
As shown in Fig. \ref{fig:1}, we consider a two-hop link of slowly-varying P2P network where an IRS is adjoined to $K$ user pairs,  where a user pair includes one ST and one DT, and each user terminal is equipped with a single antenna.
Let ${\cal S}:=\{s_1, s_2, \ldots, s_K\}$ and ${\cal D}:=\{d_1, d_2, \ldots, d_K\}$ be the sets of STs and DTs, respectively.
The index set of user pairs is denoted by ${\cal K}:=\{1, 2, \ldots, K\}.$
{The modular IRS structure as shown in Fig. \ref{fig:1} is composed of multiple modules, each of which is attached with a smart controller, where all the controllers can communicate with each other in a point-to-point fashion via fiber links.
This setting can be regarded as a generalization of the IRS architecture introduced in \cite{Wu2019Towards}.
We suppose that the entire IRS consists of $M$ independent and controllable modules, where each module has $L$ reflecting elements, and thus, the number of the total reflecting elements at IRS is $N=ML.$
}
Let ${\cal L}:=\{1, 2, \ldots, L\}$ be the index set of reflecting elements of each module.
Define ${\cal M}:=\{1, 2, \ldots, M\}$ as the set of modules at IRS.
In the following, we use $m$ to denote the index of a module or controller within ${\cal M}$ (i.e., $m\in{\cal M}$).
The channels of two-hop communications are assumed to experience quasi-static block fading, i.e., the channel coefficients from the STs to the IRS and the IRS to the DTs remain constant during each time slot, but may vary from one to another \cite{Feng2013Device}.
Let ${\mathbf h}_{k,m}\in{\mathbb C}^{L\times 1 }$ and ${\mathbf g}_{m,k}\in{\mathbb C}^{L\times 1}$ denote the uplink channel vector from ST $k$ to the $m\text{th}$ module of IRS
and the downlink channel vector from  module $m$ to DT $k$, respectively.
The associated passive beamfromer at the $m\text{th}$ module at IRS denoted by ${\pmb\Phi}^m=\text{diag}[\phi_{(m-1)L+1}, \ldots, \phi_{(m-1)L+l}, \ldots, \phi_{mL}]\in{\mathbb C}^{L\times L},$ where $\phi_{(m-1)L+l}$ is the $l\text{-th}$ entry of  reflecting coefficient matrix at module $m$, $\forall m\in{\cal M}, l\in{\cal L}.$
We assume that all the modules can potentially serve the user pairs transmitting.
Note that if all modules are triggered to serve the ST-DT communications, the problem becomes a special case which is simpler to solve.
The passive beamformer at IRS denoted by $\pmb\Phi\in{\mathbb C}^{ N\times N},$ and the associated channel from ST $k$ to the IRS and the  channel from the IRS to DT $k$, denoted by ${\mathbf h}_k\in{\mathbb C}^{N\times 1}$ and ${\mathbf g}_k\in{\mathbb C}^{N\times 1},$ respectively, are expressed as follows:
\begin{subequations}\label{s:1}
\begin{align}
{\pmb\Phi}&:=\text{bldg}\{{\pmb\Phi}^1, {\pmb\Phi}^2, \ldots, {\pmb\Phi}^M\},\\
{\mathbf h}_k&:=[ ({\mathbf h}_{k,1})^T, ({\mathbf h}_{k,2})^T, \ldots, ({\mathbf h}_{k,M})^T ]^T, \forall k\in{\cal K},\\
{\mathbf g}_k&:=[({\mathbf g}_{1, k})^T, ({\mathbf g}_{2,k})^T, \ldots, ({\mathbf g}_{M,k})^T  ]^T, \forall k\in{\cal K}.
\end{align}
\end{subequations}
{In addition, in the following analysis--for the sake of simplicity--we consider that the reflecting coefficient $\phi_n=\beta_ne^{j\theta_n}$ with continuous phase shift and continuous amplitude attenuation as ${\cal X}\triangleq\{\phi_{(m-1)L+l}: |\phi_{(m-1)L+l}|\leq 1, \forall m\in{\cal M}, l\in{\cal L}\}.$}
In practice, the feasible set of reflecting coefficient ${\cal X}$ might be more complicated, such as in \cite{Di2019Hybrid,Huang2019Reconfigurable, Guo2019Weighted}, but this feature is beyond the focus of this paper.
The main purpose of this paper is to realize the \textit{true} RRM by introducing the modular IRS structure.
Specifically, for the formulated max-min SINR problem under module size constraint, the \textit{true} RRM can be implemented via the triggered module subset identification.

Let $z_k$ denote the data symbol of ST $k$ and write $p_k$ for its corresponding power. The signal received at DT $k$ via IRS-aided link is expressed by
\begin{equation}\label{s:2}
\begin{aligned}
y_{k}&={\mathbf g}_{k}^{\dag}{\pmb\Phi}\sum_{k=1}^K\sqrt{p_{k}}{\mathbf h}_{k}z_k+u_k\\
&={\mathbf g}_{k}^{\dag}{\pmb\Phi}\sqrt{p_k}{\mathbf h}_{k}z_k+{\mathbf g}_{k}^{\dag}{\pmb\Phi}\sum_{j=1, j\neq k}^K\sqrt{p_{j}}{\mathbf h}_{j}z_j+u_k, \forall k\in {\cal K},
\end{aligned}
\end{equation}
where $u_k\sim {\cal C}{\cal N}(0, \sigma^2)$ is the thermal noise experienced by DT $k$ and the second term accounts for the interference experienced by ST-DT pair $k$ from other user pairs $j\in{\cal K}, j\neq k.$
Then, based on (\ref{s:2}) the received SINR at  $d_k$ is given by
\begin{equation}\label{eq:2}
\text{SINR}_k=\frac{p_k|{\mathbf g}_k^{\dag}{\pmb\Phi}{\mathbf h}_k|^2}
{\sum_{j=1,j\neq k}^Kp_j|{\mathbf g}_k^{\dag}{\pmb\Phi}{\mathbf h}_j|^2+\sigma^2}, \forall k\in{\cal K}.
\end{equation}
The design problem is to maximize the minimum SINR at all DTs, while satisfying the STs' transmit power constraint $p_k^{\max}, \forall k\in{\cal K},$ and the reflecting coefficient constraint ${\cal X};$ that is
\begin{align}
&\max_{\left\{{\pmb\Phi}, \{ p_k\}_{k=1}^K\right\}} \min_{k}~ \text{SINR}_k \label{add:3-1}\\
&\text{s.t.~} p_k\leq p_k^{\max}, \text{~and~} {\cal X},\forall k\in {\cal K}.\label{add:3-2}
\end{align}
The optimization problem (\ref{add:3-1})-(\ref{add:3-2}) is  non-convex, due to the non-convex objective function w.r.t. $\{p_k\}_{k=1}^K$ and ${\pmb\Phi}$, and there is no known efficient and standard method to obtain an optimal solution for it.
This motivates the pursuit of effective method for a suboptimal solution via exploiting the special structure of problem itself.
The most suitable tool for this similar problems is the generalized fractional programs \cite{Crouzeix1985An, Borde1987Convergence, Benadada1988Partial}, since the numerator and denominator of SINR are continuous functions on variables.
For this IRS-aided two-hop communication, the suboptimal solution of (\ref{add:3-1})-(\ref{add:3-2}) can be obtained effectively by using {generalized fractional programs} and the alternating optimization technique \cite{Fukushima1992Application} to separately and iteratively solve for $\{p_k\}_{k=1}^K$ and ${\pmb\Phi}$ \cite{Wu2019Beamforming,Huang2019Reconfigurable, Di2019Hybrid, Wu2019Towards}.

\subsection{Triggered  Module  Subset Identification}

In this section, suppose now that only $Q\leq M$  modules are available, and thus only $QL$ reflecting elements can be serving the user pairs simultaneously.
From \cite{Mehanna2013Joint}, the design problem is to jointly select the best $Q$ out of $M$ modules, and design the transmit power $\{p_k\}_{k=1}^K$ and the corresponding passive beamformer so that the minimum SINR among DTs is maximized, subject to the power budget at each ST and reflecting coefficient constraint.

Define the $N\times 1$ vector ${\pmb\phi}:=[({\pmb\phi}^1)^T, ({\pmb\phi}^2)^T, \ldots, ({\pmb\phi}^M)^T]^T$, where ${\pmb\phi}^m:=[\phi_{(m-1)L+1}, \ldots, {\phi}_{mL}]^T\in{\mathbb C}^{L\times 1}$ is the $m\text{th}$ block of vector ${\pmb\phi}$, $ \forall m\in{\cal M}.$
If module $m\in{\cal M}$ is not triggered, vector ${\pmb\phi}^m$ must be set to zero, i.e., ${\pmb\phi}^m={\mathbf 0}.$
Hence, the max-min SINR problem via joint triggered modules identification, power allocation, and passive beamformer design can be expressed by
\begin{align}
\text{(P0)} &\max_{\left\{{\pmb\Phi}, \{ p_k\}_{k=1}^K\right\}} \min_{k}~ \text{SINR}_k \label{eq:3-1}\\
\text{s.t.~} &{{\cal S}(\text{module})\leq Q, } \text{~and~} (\ref{add:3-2}),\label{eqq:3-3}
\end{align}
{where ${\cal S}(\text{module})$ denotes the number of triggered modules,  and  $Q\leq M$ is the upper bound of this number.}
Note that (P0) is an NP-hard problem due to the non-convex constraint ${\cal S}(\text{module})\leq Q$, and solving (P0) requires an exhaustive combinatorial search over all $\binom{Q}{M} $ possible patterns.
Thus, in the following, we aim to develop computationally efficient method to obtain a sub-optimal solution.
Specifically, instead of the hard module size constraint ${\cal S}(\text{module})\leq Q$, a sparsity-inducing approximation \cite{Mehanna2013Joint, Candes2008Enhancing} can be employed to obtain a convex relaxation of (P0).


\subsection{Convex Relaxation and Problem Formulation}

Define ${\mathbf A}_{k}=\text{diag}[{\mathbf g}_k^{\dag}]\in {\mathbb C}^{N\times N},$
the $N\times 1$ vector $\bar{\mathbf h}_{j, k}={\mathbf A}_{k}{\mathbf h}_{j},$
and the $(KN)\times 1$ vector $\bar{\mathbf h}^k=[(\bar{\mathbf h}_{1,k})^T, (\bar{\mathbf h}_{2, k})^T, \ldots,(\bar{\mathbf h}_{k,k})^T, \ldots, (\bar{\mathbf h}_{K,k})^T]^T.$
Thus, the expression of $\text{SINR}_k$ in (\ref{eq:2}) can be rewritten as
\begin{equation}\label{s:13}
\begin{aligned}
\text{SINR}_{k}&=\frac{p_{k}{\pmb\phi}^{\dag}\bar{\mathbf h}_{k,k}\bar{\mathbf h}_{k,k}^{\dag}{\pmb\phi}}
{\sigma^2+\sum_{j=1, j\neq k}^Kp_{j}{\pmb\phi}^{\dag}\bar{\mathbf h}_{j,k}\bar{\mathbf h}_{j,k}^{\dag}{\pmb\phi}}\\
&=\frac{\bar{\pmb\phi}_k^{\dag}\bar{\mathbf h}_{k,k}\bar{\mathbf h}_{k,k}^{\dag}\bar{\pmb\phi}_k}
{\sigma^2+\sum_{j=1, j\neq k}^K\bar{\pmb\phi}_j^{\dag}\bar{\mathbf h}_{j,k}\bar{\mathbf h}_{j,k}^{\dag}\bar{\pmb\phi}_j},
\end{aligned}
\end{equation}
where $\bar{\pmb\phi}_k=\sqrt{p_k}{\pmb\phi}, \forall k\in{\cal K}. $
Define the $N\times K$ matrix $\bar{\pmb\Phi}=[\bar{\pmb\phi}_1, \bar{\pmb\phi}_2, \ldots, \bar{\pmb\phi}_K].$
{Moreover, define the $L\times K$ matrix ${\bar{\pmb\Phi}^m}:=[\sqrt{p_1}{\pmb\phi}^m, \sqrt{p_2}{\pmb\phi}^m, \ldots, \sqrt{p_K}{\pmb\phi}^m],$ where $\sqrt{p_k}{\pmb\phi}^m$ is the $m\text{-th}$ block of $\bar{\pmb\phi}_k.$
Thus, $\bar{\pmb\Phi}$ can be rewritten as $\bar{\pmb\Phi}=[(\bar{\pmb\Phi}^1)^T, (\bar{\pmb\Phi}^2)^T, \ldots, (\bar{\pmb\Phi}^M)^T]^T.$
Define  the $M\times 1$ vector $\tilde{\pmb\phi}:=[||\bar{\pmb\Phi}^1||_F, ||\bar{\pmb\Phi}^2||_F, \ldots, ||\bar{\pmb\Phi}^M||_F]^T.$
If module $m$ is not be triggered, all the entries of matrix $\bar{\pmb\Phi}^m$ must be set to zero.
This means that $\bar{\pmb\Phi}^m={\mathbf 0}$, i.e.,  the $m\text{-th}$ block of each $\bar{\pmb\phi}_k$, for all $K$ ST-DT pairs, must be set to zero simultaneously.
Hence, the hard module size constraint in (\ref{eqq:3-3}) is equivalent to
\begin{equation}\label{add:1}
\begin{aligned}
||\tilde{\pmb\phi}||_{0}\leq Q,
\end{aligned}
\end{equation}
where the $\ell_0\text{-norm}$ is the number of nonzero entries of $\tilde{\pmb\phi},$ i.e.,
$||\tilde{\pmb\phi}||_0:=\left|\left\{m: ||\bar{\pmb\Phi}^m||_F\neq 0\right\}  \right|.$
By introducing an auxiliary variable $\gamma$ and replacing the hard module size constraint, an $\ell_0\text{-norm}$ penalty can be employed to promote sparsity leading to
\begin{equation}\label{add:2}
\begin{aligned}
\max_{{\pmb\phi}, \{p_k\}_{k=1}^K, \gamma}&~~~\gamma-\alpha||\tilde{\pmb\phi}||_0\\
\text{s.t.}~&~\text{SINR}_k\geq \gamma, \text{~and~} (\ref{add:3-2}), k\in{\cal K},
\end{aligned}
\end{equation}
where $\alpha>0$ is a positive real tuning parameter that controls the sparsity of the solution, i.e., the number of triggered modules. Problem (\ref{add:2}) strikes a balance between maximizing the minimum SINR among ST-DT pairs and minimizing the number of triggered modules, where a larger $\alpha$ implies a sparser solution. Note that for any $\alpha,$ there is a corresponding $Q$ for which problems (\ref{add:2}) and (P0) yield the same sparse solution, and thus the focus is placed on (\ref{add:2}) only.
However, since the optimization problem (\ref{add:2}) is nonconvex and generally impossible to solve as its solution usually requires an intractable combinatorial search.
}
The mixed ${\ell}_{1, 2}\text{-norm}$, which was first presented in the context of the group \textit{least-absolute selection and shrinkage operator} (group \textit{Lasso}) \cite{Yuan2006Model}, the triggered module size can be effectively approximated by replacing the ${\ell}_{0}\text{-norm}$ with ${\ell}_{1, 2}\text{-norm}$, i.e.,
$||{\pmb\phi}||_{1,2}\triangleq \sum_{m=1}^M ||{\pmb\phi}^m ||_2.$
In our scenario,  the mixed convex norm $\ell_{1,F}$ of matrix \cite{Lin2016Joint} can be defined as
\begin{equation}\label{eq:4}
||\bar{\pmb\Phi}||_{1,F}=\sum_{m=1}^M||\bar{\pmb\Phi}^m||_F.
\end{equation}
Note that $||\bar{\pmb\Phi}||_{1, F}=||\tilde{\pmb\phi}||_1.$ The mixed $\ell_{1, F}\text{-norm}$ behaves as the $\ell_1\text{-norm}$ on $\tilde{\pmb\phi},$ which implies that each $||\bar{\pmb\Phi}^m||_F$ (or equivalently $\bar{\pmb\Phi}^m$) is encouraged to be set to zero, therefore inducing group-sparsity.
{From \cite{Candes2008Enhancing}, the $\ell_1\text{-norm}$ is known to offer the closest convex approximation to the $\ell_0\text{-norm},$ and the sparsity-promoting nature of $\ell_1\text{-norm}$ minimization was empirically confirmed \cite{Taylor1979Deconvolution,Santosa1986Linear},
where the objective function of (\ref{add:2}) can be replaced by $\gamma-\alpha||\bar{\pmb\Phi}||_{1, F}.$
Naturally, the hard module size constraint in program (P0) can be relaxed as
\begin{equation}\label{add:3}
\alpha\sum_{m=1}^M||\bar{\pmb\Phi}^m||_F\leq \delta,
\end{equation}
where $\delta>0$ controls the row block sparsity of $\bar{\pmb\Phi}.$}
{Applying the convex relaxation of module size constraint, (P0) is equivalent to
\begin{align}
\text{(P1)}~\max_{\bar{\pmb\phi}_k}&\min_{k} \text{SINR}_k\label{add:4-1}\\
\text{s.t.}~&(\ref{add:3-2}) \text{~and~}  (\ref{add:3}).\label{add:4-3}
\end{align}
}

{
{\em{Remark 1:}} Similar to the weighted $\ell_1$ minimization problem \cite{Candes2008Enhancing}, the weight $\alpha$ can be regarded as a free parameter in the convex relaxation, whose value can be chosen to avoid trivial solutions (i.e, $\bar{\pmb\Phi}^m={\mathbf 0}$ or $\bar{\pmb\Phi}^m\neq {\mathbf 0}, \forall m\in{\cal M}$).
Notably, in (\ref{add:3}),  both $\alpha$ and $\delta$ affect the cardinal number of the triggered module subsets together.
As mentioned before, the sparsity is negatively correlated with $\alpha$, while is positively correlated with  $\delta.$
Inspired by the iterative algorithm of redefining the weights \cite{Candes2008Enhancing}, intuitively, parameter $\alpha$ should relate inversely to $\delta,$ consequently,  $\alpha$ can be  set to $1/(\delta+0.01)$.
Our motivation for introducing $0.01$ in the $\alpha$ setting is to provide stability and ensure feasibility.
In this setting, the number of triggered modules increases as the parameter $\delta$ increases until approaching the upper bound of the module quantity.
}
{
\begin{lemma}
If $\delta\geq \frac{-0.01+\sqrt{(0.01)^2+\sqrt{16MKN\max_{k}\{p_k^{\max}\}}}}
{2},$  the triggered module size is $M$, since the module size constraint in (\ref{add:3}) is inactive.
\end{lemma}
\begin{proof}
Please to refer to Appendix A.
\end{proof}
}

{As Lemma 1 states, there is a rang of favorable parameter $\delta$ for a given system setting (i.e., $M, K, N,$ and $\{p_k^{\max}\}_{k=1}^K$), which suggests that the possibility of constructing a favorable set of $\delta$ based solely on the information about the system magnitudes.
To efficiently control the triggered module size, $\delta$ should be within the range of $\left(0, \frac{-0.01+\sqrt{(0.01)^2+\sqrt{16MKN\max_{k}\{p_k^{\max}\}}}}
{2}\right).$}

By introducing an auxiliary variable $\gamma,$ the joint triggered module subset identification, transmit power allocation, and the corresponding passive beamformer design problem (P1) can thus be equivalent to
\begin{align}
\text{(P1--1)}~&\max_{{\pmb\phi}, \gamma, \{p_k\}_{k=1}^K}~\gamma  \label{eq:8}\\
\text{s.t.~}&~\text{SINR}_k\geq \gamma, \forall k\in{\cal K}\label{eq:9}\\
&(\ref{add:3-2}) \text{~and~} (\ref{add:3}). \label{eq:11}
\end{align}
It is clear from problem (P1--1) that for large $\gamma,$  (P1--1) may be infeasible due to the resulting stringent SINR constraints, strong interference, and insufficient number of triggered modules.
To this end, in the following, problem (\ref{eq:8})--(\ref{eq:11}) can be solved efficiently via bisection method for feasibility checking.
Although the feasibility checking of (P1-1) can be solved by CVX \cite{Grant2011The}, the problem dimension may be transformed to an additional challenging issue, due to the increasing collected information of all the ST-DT pairs and the IRS, when the number of them is larger.
In the next section, to develop a partially distributed algorithm, we fit (\ref{eq:8})--(\ref{eq:11}) into the ADMM framework \cite{Fukushima1992Application} and then reformulate it as a separable group \textit{Lasso} problem \cite{Yuan2006Model}.
Finally, a custom-made partially distributed algorithm is developed.
The proposed algorithm is computationally efficient since each step of ADMM can be computed in closed-form.

\section{Two-block ADMM-based Optimal Solution }
Our proposed ADMM-based solution framework is composed of two phases. The first phase is to identify the triggered module subset and the second phase is to solve the original max-min SINR problem free from the module size constraint (\ref{add:3}).
\subsection{Triggered Module  Subset Identification}
As demonstrated in Section II,  for a given $\gamma>0,$ the design problem (\ref{eq:8})--(\ref{eq:11}) becomes the feasibility test one.
In this context, the challenge in solving problem (P1--1) lies in the fact that its objective is non-differentiable and that the feasible set is nonconvex.
To proceed further, for the fixed $\gamma$, we observe that (P1--1) is feasible if and only if the solution of the following optimization problem (P1--2) is lower than $\delta,$ where (P1--2) is given by
\begin{align}
\text{(P1--2)}~ &\min_{\bar{\pmb\Phi}} ~ \sum_{m=1}^M\alpha||\bar{\pmb\Phi}^m||_F\label{eq:12}\\
\text{s.t.~}&\sqrt{(1+\gamma^{-1})}\bar{\mathbf h}_{k,k}^{\dag}\bar{\pmb\phi}_k\geq ||[\bar{\mathbf h}^{k\dag}\widetilde{\pmb\Phi}, \sigma] ||_2, \forall k\in{\cal K}\label{eq:13}\\
&\bar{\pmb\phi}_k^{\dag}{\mathbf e}_n{\mathbf e}_n^{\dag}\bar{\pmb\phi}_k\leq p_k^{\max}, n\in{\cal N}; k\in{\cal K},\label{eq:14}\\
&\text{Im}(\bar{\mathbf h}_{k,k}^{\dag}\bar{\pmb\phi}_k)=0, \forall k\in{\cal K}. \label{eq:14-1}
\end{align}
The constraint (\ref{eq:13}) of (P1--2) is the reformulation of SINR constraint (\ref{eq:9}) relies on the second-order cone program \cite{Boyd2009Convex}, where the $(NK)\times K$ matrix $\tilde{\pmb\Phi}$ is defined as $\tilde{\pmb\Phi}=\text{bldg}\{\bar{\pmb\phi}_1, \bar{\pmb\phi}_2, \ldots, \bar{\pmb\phi}_K\}.$

{Define $\widetilde{\mathbf H}=[\bar{\mathbf h}^1, \bar{\mathbf h}^2, \ldots, \bar{\mathbf h}^K]\in{\mathbb C}^{(NK)\times K}.$
To develop the performance gains brought by the passive beamforming of IRS, in the rest of this paper, it is assumed that the channels ${\mathbf h}_k$ and ${\mathbf g}_k$ are perfectly known at $s_k, \forall k\in{\cal K}$ \cite{Huang2019Reconfigurable}.
All involved channels can be estimated at IRS via their training signals, whose implementation is based on the assumption that each reflecting element is equipped with a low-power receive RF chain \cite{Wu2019Towards}. 
}
Moreover, let us introduce a $K\times (K+1)$ auxiliary matrix
${\mathbf F}=[\widetilde{\mathbf H}^{ \dag}\widetilde{\pmb\Phi}, \sigma{\mathbf 1}_{K}].$
Furthermore, define $f_{k,k}=\bar{\mathbf h}_{k,k}^{\dag}\bar{\pmb\phi}_k$ and ${\mathbf f}_k=[\bar{\mathbf h}^{k\dag}\widetilde{\pmb\Phi}, \sigma]\in{\mathbb C}^{1\times (K+1)}$ as the $k\text{-th}$ diagonal element and the $k\text{-th}$ row vector of ${\mathbf F},$ respectively. Using these definitions and introducing a matrix variable ${\mathbf W}=\bar{\pmb\Phi}\in {\mathbb C}^{N\times K}$, problem (\ref{eq:12})--(\ref{eq:14-1}) can be reformulated as the following problem:
\begin{align}
(\text{P1--3})~ &\min_{\left\{\{\bar{\pmb\phi}_k\}_{k\in{\cal K}}, {\mathbf W}, \{{\mathbf F}\}\right\}} \sum_{m=1}^{M} \alpha||{\mathbf W}^m||_F\label{eq:15}\\
\text{s. t.}& \sqrt{(1+\gamma^{-1})}f_{k,k}\geq ||{\mathbf f}_{k}||_2, \forall k\in {\cal K}\label{eq:16}\\
&\bar{\pmb\phi}_k^{\dag}{\mathbf e}_n{\mathbf e}_{n}^{\dag}\bar{\pmb\phi}_k\leq p_k^{\max}, \forall n\in{\cal N}; k\in{\cal K},\label{eq:17}\\
&{\mathbf W}=\bar{\pmb\Phi},~~~
{\mathbf F}=[\widetilde{\mathbf H}^{ \dag}\widetilde{\pmb\Phi}, \sigma{\mathbf 1}_{K}]\label{eq:18}.
\end{align}
(P1--3) is a convex minimization problem, and therefore, the duality gap between (P1--3) and its augmented duality problem is zero.
This means that the optimal solution of (P1--3) can be obtained by applying the augmented Lagrangian duality theory \cite{Hestenes1969Multiplier}.
The partial augmented Lagrangian function of (P1--3) can be written as
\begin{equation}\label{s:21}
\begin{aligned}
{ L_c}&(\{\bar{\pmb\phi}_k\}_{k\in{\cal K}}, {\mathbf W}, {\mathbf F}, {\pmb\Lambda}, {\pmb\Psi})\\
\triangleq&  \sum_{m=1}^{{M}}\alpha||{\mathbf W}^m||_F+
\text{Re}\{\text{Tr}[{\pmb\Lambda}^{\dag}({\mathbf W}-\bar{\pmb\Phi})]\}
+\frac{c}{2}||{\mathbf W}-\bar{\pmb\Phi}||_F^2\\
&+\text{Re}\{\text{Tr}[ {\pmb\Psi}^{\dag}({\mathbf F}-[\widetilde{\mathbf H}^{ \dag}\widetilde{\pmb\Phi}, \sigma{\mathbf 1}_{K}])]\}
+\frac{c}{2}||{\mathbf F}-[\widetilde{\mathbf H}^{\dag}\widetilde{\pmb\Phi}, \sigma{\mathbf 1}_{K}]  ||_F^2,
\end{aligned}
\end{equation}
where $c>0$ is the penalty factor; ${\pmb\Lambda}\in {\mathbb C}^{N\times K}$ and ${\pmb\Psi}\in{\mathbb C}^{K\times (K+1)}$ are the Lagrangian matrix multipliers for ${\mathbf W}=\bar{\pmb\Phi}$ and ${\mathbf F}=[\widetilde{\mathbf H}^{\dag}\widetilde{\pmb\Phi}, \sigma{\mathbf 1}_{K}]$, respectively.
Note that the reformulated SINR constraints  (\ref{eq:16}) as well as the boundary constraint (\ref{eq:17}) are not taken into the augmented Lagrangian function and they will be integrated into the optimal solution in the following.
Particularly, we focus on solving:
\begin{equation}\label{s:20}
\begin{aligned}
(\text{P1--4})~&\max_{\left\{{\pmb\Lambda},{\pmb\Psi}\right\}}\min_{\left\{\{\bar{\pmb\phi}_k\}_{k\in{\cal K}}, {\mathbf W}, {\mathbf F}\right\}}
{ L_c}\left(\{\bar{\pmb\phi}_k\}_{k\in{\cal K}}, {\mathbf W}, {\mathbf F}, {\pmb\Lambda}, {\pmb\Psi}\right)\\
\text{s. t. ~}& \sqrt{(1+\gamma^{-1})}f_{k,k}\geq ||{\mathbf f}_{k}||_2, \forall k\in {\cal K}\\
&\bar{\pmb\phi}_k^{\dag}{\mathbf e}_n{\mathbf e}_n^{\dag}\bar{\pmb\phi}_k\leq p_k^{\max}, \forall n\in{\cal N}; k\in{\cal K}.
\end{aligned}
\end{equation}

We need to decouple the optimization variables in $L_c$ to make (P1--4) intractable.
To be specific, dividing $\{\bar{\pmb\phi}_k\}_{k\in{\cal K}}, {\mathbf W}, $ and ${\mathbf F}$ into two blocks of $\{\bar{\pmb\phi}_k\}_{k\in{\cal K}}$ and $\{{\mathbf W}, {\mathbf F}\}$, we can apply the two-block ADMM framework \cite{Fukushima1992Application} to solve (P1--4).
In each iteration $t$, we first update $\{\bar{\pmb\phi}_k\}_{k\in{\cal K}}$ by solving ${ P}_{\bar{\pmb\phi}_k}$
\begin{equation}
\begin{aligned}
{ P}_{\bar{\pmb\phi}_k}: \min_{\bar{\pmb\phi}_k} &\text{Re}\{\text{Tr}[{\pmb\lambda}^{k\dag}({\mathbf w}^k(t)-\bar{\pmb\phi}_k) ]\}+\frac{c}{2}||{\mathbf w}^k(t)-\bar{\pmb\phi}_k ||_2^2\\
&+\text{Re}\{\text{Tr}[{\pmb\psi}^{k \dag}({\mathbf f}^{k}(t)-\bar{\mathbf H}^{k \dag}\bar{\pmb\phi}_k)]\}\\
&+\frac{c}{2}\|{\mathbf f}^{k}(t)-\bar{\mathbf H}^{k \dag}\bar{\pmb\phi}_k \|_2^2\\
\text{s.t.~}&~\bar{\pmb\phi}_k^{\dag}{\mathbf e}_n{\mathbf e}_n^{\dag}\bar{\pmb\phi}_k\leq p_k^{\max}, \forall k\in{\cal K}; n\in{\cal N},
\end{aligned}
\end{equation}
where ${\mathbf w}^k\in {\mathbb C}^{N\times 1}$ and ${\pmb\lambda}^k\in {\mathbb C}^{N\times 1}$ represent the $k\text{-th}$ column of matrices ${\mathbf W}$ and ${\pmb\Lambda}$, respectively;
and ${\pmb\psi}^{k}\in{\mathbb C}^{K\times 1}$ and ${\mathbf f}^{k}\in {\mathbb C}^{K\times 1} $
are the $k\text{-th}$ column of matrices ${\pmb\Psi}$ and ${\mathbf F}$, respectively;
$\bar{\mathbf H}^{k}=[\bar{\mathbf h}_{k,1}, \bar{\mathbf h}_{k,2}, \ldots,\bar{\mathbf h}_{k,K}]\in{\cal C}^{N\times K}.$
With the obtained $\bar{\pmb\Phi}$ (or $\{\bar{\pmb\phi}_k\}_{k\in{\cal K}}$), and then better solutions for ${\mathbf W}$ and ${\mathbf F}$ can be updated by solving the following problem:
\begin{equation}\label{eq:18}
\begin{aligned}
{ P}_{{\mathbf W}, {\mathbf F}}:&~\min_{{\mathbf W}, {\mathbf F}} L_c(\bar{\pmb\Phi}(t+1), {\mathbf W}, {\mathbf F}, {\pmb\Lambda},{\pmb\Psi})\\
\text{s.t.}& \sqrt{(1+\gamma^{-1})}f_{k,k}\geq ||{\mathbf f}_{k}||_2, \forall k\in{\cal K}.
\end{aligned}
\end{equation}
Then, as shown in Appendix B, the optimal ${\bar{\pmb\phi}_k}$ and ${\mathbf W}$ can be obtained as in Theorem 1.
\begin{theorem}
For given $\pmb\Lambda $ and  $\pmb\Psi,$ the optimal $\{\bar{\pmb\phi}_k\}_{k=1}^K$ of minimizing ${ P}_{\bar{\pmb\phi}_k}$ is given by
\begin{equation}\label{s:24-1}
\begin{aligned}
\bar{\pmb\phi}_k(t+1)=&( c{\mathbf I}_{N\times N}+c\tilde{\mathbf h}^k\bar{\mathbf H}^{k\dag}+2\sum_{n=1}^N \mu_n^k{\mathbf e}_n{\mathbf e}_n^{\dag})^{-1}\\
&\times({\pmb\lambda}^k(t)+c{\mathbf w}^k(t)+\tilde{\mathbf h}^k{\pmb\psi}^{k}(t)
+c\bar{\mathbf H}^k{\mathbf f}^{k}(t)),
\end{aligned}
\end{equation}
where  $\mu_n^k\geq 0$ is the Lagrangian multiplier of boundary constraint in (\ref{eq:17}).
Moreover, the optimal ${\mathbf W}$ is given by solving the following unconstrained problem
\begin{equation}\label{s:25}
\begin{aligned}
 \min_{\mathbf W} \sum_{m=1}^{M}& \alpha||{\mathbf W}^m||_F
+\text{Re}\{\text{Tr}[{\pmb\Lambda}^{\dag}({\mathbf W}-\bar{\pmb\Phi}(t+1))] \}\\
&+\frac{c}{2}|| {\mathbf W}-\bar{\pmb\Phi}(t+1) ||_F^2.
\end{aligned}
\end{equation}
Using the first-order optimality condition for the optimal solution ${\mathbf W}^m(t+1)$, we have
\begin{equation}\label{s:29-1}
\begin{aligned}
{\mathbf W}^m(t+1)=\left\{\begin{array}{lll}
&{\mathbf 0}, &\text{~if~} ||{\pmb\Xi}(t)||_F\leq \alpha\\
&\frac{(||{\pmb\Xi}^m(t)||_F-\alpha){\pmb\Xi}^m(t)}
{c||{\pmb\Xi}^m(t)||_F}, &\text{~otherwise},
\end{array}\right.
\end{aligned}
\end{equation}
where ${\pmb\Xi}^m(t)=c\bar{\pmb\Phi}^m(t+1)-{\pmb\Lambda}^m(t)$, and ${\pmb\Lambda}^m\in {\mathbb C}^{{L}\times K},$ ${\mathbf W}^m\in {\mathbb C}^{{L}\times K},$ and $\bar{\pmb\Phi}^m\in {\mathbb C}^{{L}\times K}$
are the $m\text{-th}$ row blocks of matrices ${\pmb\Lambda}, {\mathbf W},$ and $\bar{\pmb\Phi}$, respectively,  $\forall m\in{\cal M}.$
\end{theorem}
\begin{proof}
See Appendix B.
\end{proof}

With the obtained optimal $\{\bar{\pmb\phi}_k\}_{k\in{\cal K}}$, the optimal multiplier $\mu_n^{k}$ of the boundary constraint (\ref{eq:17}) can be optimally obtained  by
\begin{equation}
\mu_n^k=\left( p_k^{\max}-\bar{\pmb\phi}_k^{\dag}(t+1){\mathbf e}_n{\mathbf e}_n^{\dag}\bar{\pmb\phi}_k(t+1)\right)^{+}, \forall k\in{\cal K}; n\in{\cal N},
\end{equation}
where $(x)^{+}=\max\{x, 0\}.$

Finally, we optimize ${\mathbf F}$ in (\ref{eq:18}) given fixed $\bar{\pmb\Phi}.$
The problem of ${\mathbf F}$ of (\ref{eq:18}) is expressed as
\begin{equation}\label{s:30}
\begin{aligned}
{ P}_{\mathbf F}: \min_{\mathbf F}&~
 \text{Re}\{\text{Tr}[{\pmb\Psi}^{\dag}(t)({\mathbf F}-[\widetilde{\mathbf H}^{ \dag}\tilde{\pmb\Phi}(t+1), \sigma{\mathbf 1}_{K}])  ]\}\\
&~~~~~~+\frac{c}{2}|| {\mathbf F}- [\widetilde{\mathbf H}^{\dag}\widetilde{\pmb\Phi}(t+1), \sigma{\mathbf 1}_{K}]||_F^2\\
\text{s.t.}~&~\sqrt{\gamma^{-1}}f_{k,k}\geq ||{\mathbf f}_{-k,k}||_2, \forall k\in {\cal K},
\end{aligned}
\end{equation}
where ${\mathbf f}_{-k,k}\in{\mathbb C}^{1\times K}$ denotes the remaining subvector of ${\mathbf f}_{k}\in{\mathbb C}^{1\times (K+1)}$ after removing the element $f_{k,k},$ i.e.,
${\mathbf f}_{-k,k}=[f_{1,k},  \ldots,f_{k-1,k},f_{k+1, k}, \ldots, f_{K+1,k}]\in{\mathbb C}^{1\times K}.$
Similar to \cite{Lin2016Joint}, in order to find the optimal ${\mathbf F}$ for ${ P}_{\mathbf F}$ with low computational complexity, we divide the optimization problem of ${\mathbf F}$ into $K$ independent subproblems of ${\mathbf f}_k, \forall k\in{\cal K},$ which are solved in parallel, where the subproblem is given by
\begin{equation}\label{s:31}
\begin{aligned}
&\min_{{\mathbf f}_k}
\text{Re}\left\{\text{Tr}\left[{\pmb\psi}_k(t)^{\dag}\left({\mathbf f}_k-{\mathbf b}_k(t+1) \right)\right]\right\}+\frac{c}{2}|| {\mathbf f}_k-{\mathbf b}_k(t+1) ||_2^2\\
&\text{s.t. ~} \sqrt{\gamma^{-1}} f_{k,k}\geq ||{\mathbf f}_{-k,k} ||_2,
\end{aligned}
\end{equation}
where ${\pmb\psi}_k\in{\mathbb C}^{1\times (K+1)}$ and $ {\mathbf b}_k\in{\mathbb C}^{1\times(K+1)}$ denote the $k\text{-th}$ row vectors of ${\pmb\Psi}$ and $[\widetilde{\mathbf H}^{\dag}\widetilde{\pmb\Phi}, \sigma{\mathbf 1}_{K}],$ respectively.
Define $\psi_{k,k}$ as the $k\text{-th}$ element of ${\pmb\psi}_k$ and ${\pmb\psi}_{-k,k}$ as the remaining subvector after removing $\psi_{k,k}.$
Similarly, we define $b_{k,k}$ and ${\mathbf b}_{-k,k}$.

Problem (\ref{s:31}) is a convex minimization problem. Moreover, it can be verified that the Slater's
constraint qualification is satisfied \cite{Boyd2009Convex}.
Therefore, the duality gap between problem (\ref{s:31}) and its duality problem is zero.
This means that the optimal solution of problem (\ref{s:31}) can be obtained by applying the Lagrange duality theory \cite{Boyd2009Convex}. In the following Theorem 2,  the optimal ${\mathbf f}_k$ can be obtained via exploiting the Karush-Kuhn-Tucker (KKT) conditions of problem (\ref{s:31}).
\begin{theorem}
Given ${\pmb\Psi},$ fixing $\bar{\pmb\Phi}(t+1),$ the optimal ${\mathbf f}_k$ is
\begin{equation}\label{s:33-1}
\left\{\begin{array}{ll}
&f_{k,k}(t+1)=\frac{cb_{k,k}(t+1)-\psi_{k,k}(t)+\sqrt{\gamma^{-1}}\varepsilon_k}{c}\\
&{\mathbf f}_{-k,k}(t+1)=\frac{c{\mathbf b}_{-k,k}(t+1)-{\pmb\psi}_{-k,k}(t)}
{c+\varepsilon_k\rho_k},
\end{array}\right.
\end{equation}
where $\rho_k=(||{\mathbf f}_{-k,k}(t+1)||_F)^{-1}$, $\varepsilon_k\geq 0$ is the dual variable introduced for the SINR constraint, which is optimally determined by
\begin{equation}\label{sss:1}
\begin{aligned}
\varepsilon_k=&\frac{1}{1+\gamma}[\gamma||c{\mathbf b}_{-k,k}(t+1)-{\pmb\psi}_{-k,k}(t) ||_2\\
&-\sqrt{\gamma}(cb_{k,k}(t+1)-\psi_{k,k}(t))].
\end{aligned}
\end{equation}
\end{theorem}
\begin{proof}
See Appendix C.
\end{proof}

After obtaining the optimal $\bar{\pmb\Phi},$ ${\mathbf W},$ and ${\mathbf F},$ we update the Lagrangian matrix multipliers in problem (P1--4), i.e., ${\pmb\Lambda}$ and $\pmb\Psi$.
It has been shown in \cite{Ramamonjison2015Energy} that the well known subgradient based method can be employed iteratively to find the optimal solutions of ${\pmb\Lambda}$ and ${\pmb\Psi}$.
Similar to the updating of variables $\{{\mathbf W}^m\}_{m=1}^M$ and $\{{\mathbf f}_k\}_{k=1}^K,$ updating ${\pmb\Lambda}$ and ${\pmb\Psi}$ are also separable.
Specifically, for ${\pmb\Lambda}^m$ and ${\pmb\psi}_k$, the pointwise update equations are given by
\begin{align}
{\pmb\Lambda}^m(t+1)&={\pmb\Lambda}^m(t)+c\left({\mathbf W}^m(t+1)-\bar{\pmb\Phi}^m(t+1)\right)
\label{s:34-1}\\
{\pmb\psi}_k(t+1)&={\pmb\psi}_k(t)+c({\mathbf f}_k(t+1)-[\bar{\mathbf H}^{k\dag}\widetilde{\pmb\Phi}(t+1), \sigma]).\label{s:34-2}
\end{align}

\subsection{Max-Min SINR Optimization}

In this section, we solve the joint transmit power allocation and passive beamformer design problem when the triggered module subset is identified.
More precisely,  for the original max-min SINR  problem (P1), the module size constraint (\ref{add:3}) is dropped and the diagonal blocks of $\pmb\Phi$ corresponding to the non-triggered modules are forced to be zero.
For convenience to illustrate, the phase-shift matrix with identified triggered modules denoted by ${\cal F}_{\pmb\Phi}.$
Particularly, we focus on solving:
\begin{equation}\label{eq:20}
\begin{aligned}
\max_{\{p_k\}_{k\in{\cal K}}, {\cal F}_{\pmb\Phi}}\min_{k\in{\cal K}}& \frac{p_k|{\mathbf g}_k^{\dag}{\cal F}_{\pmb\Phi}{\mathbf h}_k|^2}
{\sum_{j=1,j\neq k}^Kp_j|{\mathbf g}_k^{\dag}{\cal F}_{\pmb\Phi}{\mathbf h}_j|^2+\sigma^2},
\text{~s.t.~}(\ref{add:3-2})
\end{aligned}
\end{equation}
Note that (\ref{eq:20}) can be efficiently and optimally solved by employing the alternating and optimization technique \cite{Csiszar1984Information} to separately and iteratively solve for $\{p_k\}_{k\in{\cal K}}$ and ${\cal F}_{\pmb\Phi}.$
In the rest of this section, the optimization with respect to ${\cal F}_{\pmb\Phi}$ for fixed $\{p_k\}_{k\in{\cal K}}$, and with respect to $\{p_k\}_{k\in{\cal K}}$ for fixed ${\cal F}_{\pmb\Phi}$ will be treated separately.

\subsubsection{Optimizing Phase-Shift Matrix ${\cal F}_{\pmb\Phi}$}
Let ${\cal F}_{\pmb\phi}\in{\mathbb C}^{N\times N}$ denote the vectorization of diagonal matrix ${\cal F}_{\pmb\Phi}.$ Substituting $\bar{\mathbf h}_{j,k}$ into the objective function of (\ref{eq:20}), then, ${p_k|{\mathbf g}_k^{\dag}{\cal F}_{\pmb\Phi}{\mathbf h}_k|^2}=p_k{{\cal F}_{\pmb\phi}}^{\dag}\bar{\mathbf h}_{k,k}\bar{\mathbf h}_{k,k}^{\dag}{\cal F}_{\pmb\phi},$
${\sum_{j=1,j\neq k}^Kp_j|{\mathbf g}_k^{\dag}{\cal F}_{\pmb\Phi}{\mathbf h}_j|^2+\sigma^2}=\sum_{j=1,j\neq k}^{K}p_j{\cal F}_{\pmb\phi}^{\dag}\bar{\mathbf h}_{j,k}\bar{\mathbf h}_{j,k}^{\dag}{\cal F}_{\pmb\phi}+\sigma^2$ for all $k$ and $j.$
{
Therefore, for a fixed transmit power allocation $\{p_k\}_{k=1}^K$,  problem (\ref{eq:20}) can be transformed into the following problem:
\begin{equation}\label{add:5}
\begin{aligned}
\max_{{\cal F}_{\pmb\phi}}\min_k \frac{p_k{{\cal F}_{\pmb\phi}}^{\dag}\bar{\mathbf h}_{k,k}\bar{\mathbf h}_{k,k}^{\dag}{\cal F}_{\pmb\phi}}
{\sum_{j=1,j\neq k}^{K}p_j{\cal F}_{\pmb\phi}^{\dag}\bar{\mathbf h}_{j,k}\bar{\mathbf h}_{j,k}^{\dag}{\cal F}_{\pmb\phi}+\sigma^2},
~~~\text{s.t.}~{\cal X}.
\end{aligned}
\end{equation}}
\begin{algorithm}[!t]
\caption{ Max-min SINR Optimization When the Triggered Module Subset is Identified}
\label{alg:1}
\begin{algorithmic}[1] 
\STATE Initialize  ${\cal F}_{\pmb\phi}^{(0)} $ and ${\mathbf p}^{(0)}$ to feasible values; initialize the upper bound $\bar{\gamma}_{\text{in}}=\bar{\gamma}_{\text{out}}$ and the lower bound $\underline{\gamma}_{\text{in}}=\underline{\gamma}_{\text{out}}$ of SINR in bisection,
and set the iteration number $\tau=0.$

{\bf repeat}

\STATE  For given ${\mathbf p}^{(\tau)},$ update $\gamma_{\text{out}}=\frac{\bar{\gamma}_{\text{out}}+\underline{\gamma}_{\text{out}}}{2},$ solve (\ref{add:6-1})-(\ref{add:6-3}) by CVX,

\hspace*{0.5cm}{\bf if } \texttt{CVX Status is Solved}, update $\underline{\gamma}_{\text{out}}=\gamma_{\text{out}};$

\hspace*{0.5cm}{\bf else}, \texttt{CVX Status is NAN}, update $\bar{\gamma}_{\text{out}}=\gamma_{\text{out}};$

\STATE {\bf if} $\gamma_{\text{out}}$ converges, i.e., $|\bar{\gamma}_{\text{out}}-\underline{\gamma}_{\text{out}}|\leq \epsilon,$
denote $\gamma_{\text{out}}^{(\tau+1)}, {\cal F}_{\pmb\phi}^{(\tau+1)}$ be the  optimal solution;

\STATE For given ${{\cal F}_{\pmb\phi}}^{(\tau)},$ update $\gamma_{\text{in}}=\frac{\bar{\gamma}_{\text{in}}+\underline{\gamma}_{\text{in}}}{2},$ solve (\ref{add:7}) by CVX,

\hspace*{0.5cm}{\bf if } \texttt{CVX Status is Solved}, update $\underline{\gamma}_{\text{in}}=\gamma_{\text{in}};$

\hspace*{0.5cm}{\bf else}, \texttt{CVX Status is NAN}, update $\bar{\gamma}_{\text{in}}=\gamma_{\text{in}};$

\STATE {\bf if} $\gamma_{\text{in}}$ converges, i.e., $|\bar{\gamma}_{\text{in}}-\underline{\gamma}_{\text{in}}|\leq \epsilon,$
denote $\gamma_{\text{in}}^{(\tau+1)}, {\mathbf p}^{(\tau+1)}$ be the optimal solution;

\STATE  {\bf until}  $|\gamma_{\text{out}}-\gamma_{\text{in}}|\leq \epsilon$. 
\end{algorithmic}
\end{algorithm}
{
By introducing parameter $\gamma_{\text{out}}$, problem (\ref{add:5}) is equivalent to solving problem (\ref{add:6-1})-(\ref{add:6-3}), which can be solved via a bisection procedure of feasibility checking
\begin{align}
\max_{{\cal F}_{\pmb\phi},\gamma_{\text{out}}}&~\gamma_{\text{out}}\label{add:6-1}\\
\text{s.t.}~&\sqrt{p_k(1+\gamma_{\text{out}}^{-1})}\bar{\mathbf h}_{k,k}^{\dag}{\cal F}_{\pmb\phi}\geq \nonumber\\
& \| [{\cal F}_{\pmb\phi}^{T}\bar{\mathbf h}_k^{\dag}\text{diag}\{\sqrt{p_1}{\mathbf 1}_N, \ldots, \sqrt{p_K}{\mathbf 1}_N\}, \sigma]\|_2, \forall k\in{\cal K}\label{add:6-2}\\
&\text{Im}(\bar{\mathbf h}_{k,k}^{\dag}{\cal F}_{\pmb\phi})=0, \forall k\in{\cal K}; \text{~and~} {\cal X},\label{add:6-3}
\end{align}
where $\left[{\cal F}_{\pmb\phi}^{T}\bar{\mathbf h}_k^{\dag}\text{diag}\{\sqrt{p_1}{\mathbf 1}_N, \ldots, \sqrt{p_K}{\mathbf 1}_N\}, \sigma\right]$ is the $1\times (K+1)$ row vector.
The feasibility checking of (\ref{add:6-1})-(\ref{add:6-3}) can be solved directly, e.g., by CVX, since it is a second-order cone program.}

\subsubsection{Optimization with Respect to the Power Allocation $\{p_k\}_{k\in{\cal K}}$}
{For the case where ${\cal F}_{\pmb\phi}$ is fixed and the objective is the optimization over ${\mathbf p}=[p_1, p_2, \ldots, p_K]^T,$ consequently, we focus on our attention in the following optimization problem
\begin{equation}\label{add:7}
\begin{aligned}
\max_{\{p_k\}_{k=1}^K}&\min_{k} \frac{p_k|\bar{\mathbf h}_{k,k}^{\dag}{\cal F}_{\pmb\phi}|^2}
{\sum_{j=1,j\neq k}^Kp_j|\bar{\mathbf h}_{j,k}^{\dag}{\cal F}_{\pmb\phi} |^2+\sigma^2},
\text{~s.t.~} p_k\leq p_k^{\max}.
\end{aligned}
\end{equation}
Likewise, the above max-min SINR problem is equivalent to the following problem, which can be solved also via a bisection procedure of feasibility checking
\begin{equation}\label{add:8}
\begin{aligned}
&\max_{\{p_k\}_{k=1}^K,\gamma_{\text{in}}}~\gamma_{\text{in}}\\
\text{s.t.}~& \text{SINR}_k\geq \gamma_{\text{in}},  \text{~and~} p_k\leq p_k^{\max}.
\end{aligned}
\end{equation}
}

In the proposed alternating optimization algorithm, we solve ${\mathbf p}$ and ${\cal F}_{\pmb\phi}$ by addressing problems (\ref{add:6-1})-(\ref{add:6-3}) and (\ref{add:8}) alternately in an iterative manner, where the solution obtained in each iteration
is used as the initial point of the next iteration. The details of the proposed algorithm are
summarized in Algorithm \ref{alg:1}.

Based on the analysis of the triggered module subset identification and the general max-min SINR optimization problem, we present the two-block ADMM algorithm shown in Algorithm \ref{alg:Framwork} for solving the considered optimization problem (P1-1).

\begin{algorithm}[!t]
\caption{ Two-block ADMM Algorithm for (P1-1) }
\label{alg:Framwork}
\begin{algorithmic}[1] 
\STATE Initialization:

Input communication system configurations and algorithm parameters;

{\em{phase 1: Triggered module subset identification }}

\STATE Set outer (bisection) iteration index $\tau=0$;

\STATE Update $\gamma=\frac{\underline{\gamma}+\bar{\gamma}}{2},$ where $\underline{\gamma}$ and $\bar{\gamma}$ are the lower bound and upper bound of SINR in bisection.

\STATE Set inner (ADMM) iteration index $t=0$;

\STATE  Initialize $ \{\bar{\pmb\phi}_k(t)\}_{k\in{\cal K}}, {\mathbf W}(t), \{{\mathbf F}(t)\}, {\pmb\Lambda}(t),$ and $\{{\pmb\Psi}(t)\}$;

\STATE Update $\bar{\pmb\phi}_k(t+1)$ as (\ref{s:24-1});
  ${\mathbf W}^m(t+1)$ as (\ref{s:29-1});
   ${\mathbf f}_k(t+1)$ as (\ref{s:33-1});
   ${\pmb\Lambda}^{m}(t+1)$ as (\ref{s:34-1});  ${\pmb\psi}_k(t+1)$ as (\ref{s:34-2}) in parallel; $\forall k=1, 2, \ldots, K$; $m=1, 2, \ldots, {M}$;
   
\STATE {\bf if} not converge and max iteration number not achieved, 

 \hspace*{1.5cm} $t=t+1,$ {\bf go to} (5);
 
 \hspace*{2mm} {\bf else if }convergence, compare $\sum_{m=1}^{{M}}\alpha||{\mathbf W}^m||_F$ with $\delta$,
 
 \hspace*{0.4cm} {\bf if } $``\leq"$, (P1--4) feasible for $\gamma,$ update $\underline{\gamma}=\gamma$;
 
 \hspace*{0.4cm} {\bf else} (P1--4) infeasible for $\gamma$, update $\bar{\gamma}=\gamma$;
 
\STATE  {\bf if} $\gamma$ converges, {\bf go to} (8)

 \hspace*{10mm} {\bf else}, $\tau=\tau+1,$ {\bf go to } (2);
 
\STATE Identify the triggered modules at IRS  by exploring the sparse pattern of ${\mathbf W};$

  {\em{phase 2: Transmit power allocation and passive beamforming design}}

\STATE Run Alg. \ref{alg:1} for (\ref{eq:20}) to obtain ${\mathbf p}$ and ${\cal F}_{\pmb\phi}.$
\end{algorithmic}
\end{algorithm}

\begin{table*}[!t]
\centering  \caption{Simulation Parameters}
\begin{tabular}{|p{0.42\textwidth}|p{0.1\textwidth}||p{0.25\textwidth}|l|}
\hline
Maximum transmit power of ST, $p^{\max}$ & {$20\text{~dBm}$}&
Noise power, $\sigma^2$ & $ -90\text{~dBm}$\\ \hline
Bandwidth, ${\cal W}$ & $10\text{~MHz}$&
Carrier frequency & $2.3\text{~GHz}$\\  \hline
{The number of reflection elements of each module, $L$}  & {$20$}&
Weighted factor, $\alpha$ & $1/(\delta+0.01)$ \\\hline
Circuit dissipated power at each ST, $P_{\text{ST}}$&10\text{~dBm}& Dissipated power at each DT, $P_{\text{DT}}$& $10\text{~dBm}$\\ \hline
Circuit dissipated power coefficients at ST $\xi_{\text{ST}}$ and AF relay $\xi_{\text{AF}}$&{$1.2$}& {Dissipated power at each module, $P(L)$}&{$(L\cdot 0.01)$ W}\\ \hline
\end{tabular}
\label{tab:1}
\end{table*}

\subsection{{Algorithm Summary}}
{Algorithm \ref{alg:Framwork} consists of two phases: (i) the {\em{triggered module subset identification}} phase (lines 2--9) and (ii) the {\em{transmit power allocation and passive beamformer design}} phase (line 10).
Except for the update of $\{\bar{\pmb\phi}_k\}_{k\in{\cal K}}$, the two-block ADMM algorithm can be implemented  parallelly and distributedly for triggered module subset identification (i.e., ${\mathbf W}$).
For our scenario, Alg. 2 starts by the  processor (it can communicate with all users and IRS' controllers) collecting the CSI of $\{{\mathbf h}_k\}_{k=1}^K$ and $\{{\mathbf g}_k\}_{k=1}^K$; all the auxiliary variables ${\mathbf W}$ and $ {\mathbf F}$ and the corresponding Lagrangian multipliers ${\pmb\Lambda}$ and ${\pmb\Psi}$ are initialized.
In each iteration, the processor updates $\{\bar{\pmb\phi}_k\}_{k=1}^K$  and sends $\bar{\pmb\Phi}^m$ and $\bar{\mathbf h}_{k,k}^{\dag}\bar{\pmb\phi}_k$ to the $m\text{-th}$ module and $S_k$, respectively.
In the experiments,  each IRS module $m$ executes $\{{\mathbf W}^m, {\pmb\Lambda}^m\}$ and sends it to the processor.
The $K$ STs update $\{{\mathbf f}_k^{\dag}, {\pmb\psi}_k^{\dag}\}_{k=1}^{K}$ and send it back to the processor.
Note that the difference between the two-block ADMM algorithm and the centralized algorithm is that there is no information exchange during solving (\ref{eq:12})-(\ref{eq:14}) for the latter.
In fact, in the centralized algorithm, all the calculations are completed in the central controller.
Note that as each inner iteration of Alg. 2  requires only solving convex optimization problems, hence the overall complexity of Alg. 2 is polynomial in the worst scenario.
Since the sets of solution problems (P1--3) and (P1--4) are nonempty (at least one zero solution), and
the set of solution of (P1--3) is closed and bounded, Alg. 2 is guaranteed to converge to the global optimum [30].
}

\section{Simulation Results}

\subsection{Simulation Environments and Settings}
We evaluate the performance of the proposed joint design of triggered module subset identification, transmit power allocation, and the corresponding passive beamformer in the IRS-aided P2P networks.
The convergence property and effectiveness of the two-block ADMM algorithm are verified.
We consider the IRS-aided cooperative communication system consisting with $K$ ST-DT pairs and an IRS with $N=ML$ reflecting elements, where $M$ is the number of IRS modules and $L$ is the number of reflecting elements of each module.
{Suppose that the $K$ STs are randomly and uniformly deployed within a circle cell centered at $(0,0)~\text{m}$ with the cell radius $2~\text{m}$,
and the corresponding $K$ DTs are located within a circle cell with radius $2~\text{m}$ centered at $(200, 0)~\text{m}.$
Besides, the IRS is assumed to be fixed at the location $(120, 50)\text{~m}.$
Unless specified otherwise, the simulation setting (outlined in Table \ref{tab:1}) is given as follows.
The number of ST-DT pairs is less than or equal to the number of modules at the IRS, i.e, $K\leq M.$
We consider an IRS-aided communication system with carrier frequency  $2.3\text{~GHz}$ and a system bandwidth ${\cal W}=10\text{~MHz}.$
From \cite{Zheng2020Intelligent}, we set the path loss exponent of the ST-DT pair direct link as $3.5,$   and the path loss at the reference distance  $1\text{~m}$ is set as $30\text{~dB}$ for each individual link \cite{Wu2019Intelligentjournal, Guo2019Weighted}.
For the IRS-aided link, $2$ and $2.1$ are the values of the path loss exponents from STs to the IRS and that from the IRS to DTs, respectively.
Moreover, the path loss model for the NLOS paths is characterized by Rayleigh fading.
Channel vectors $\{{\mathbf h}_{k}\}$ and $\{{\mathbf g}_{k}\}$ are  generated as i.i.d. zero-mean complex Gaussian random vectors, where the variance of each channel is determined using pathloss model $\sigma_{{\mathbf h}_{k}}^2=(200/d_{{\mathbf h}_{k}})^2$ with $d_{{\mathbf h}_{k}}$ as the distance between ST $k$ and IRS \cite{Lin2016Joint}.
Likewise, $\{{\mathbf g}_{k}\}$ can be generated according to the distribution ${\cal C}{\cal N}(0,\sigma_{{\mathbf g}_{k}}^2)$, where the variance is given by $\sigma_{{\mathbf g}_{k}}^2=(200/d_{{\mathbf g}_{k}})^{2.1}$ with $d_{{\mathbf g}_{k}}$ being the distance between IRS and DT $k.$
We assume quasi-static block fading channels in this paper, i.e., the channels from the STs to the IRS and that from the IRS to the DTs remain constant during each time block, but may vary from one to another \cite{Yang2019Low}.
For the proposed two-block ADMM, the convergence tolerance is $\epsilon=10^{-4}.$
For simplicity, all the STs are assumed to have the same maximum transmit power, i.e., $p_k^{\max}=p^{\max}=20\text{~dBm}$ and the noise power at all the destination terminals is assumed to be identical with $\sigma^2=-90\text{~dBm}.$
The number of reflecting elements of each module  is $L=20$  and each user terminal is equipped with a single antenna.
All the simulation results are obtained by averaging over $10^4$ channel realizations. Throughout the simulations, unless otherwise specified, we adopt the parameters reported in Table \ref{tab:1} (see \cite{Huang2019Reconfigurable,Zheng2020Intelligent} and references therein).
}
\begin{table*}[!t]
\centering
\caption{Performance of the Two-block ADMM Algorithm and the CVX Method for Different Sparsity Constraints $\delta \in\{4.5, 5, 5.5, 6\}$ with $K=\{5, 10\}$.  }
\resizebox{\textwidth}{10mm}{
\begin{tabular}{|l|c|c|c|c|}
\hline
~&\multicolumn{2}{|c|}{Two-block ADMM}& \multicolumn{2}{|c|}{CVX}\\ \hline
~&$K=5, M=10, L=20$&$K=10, M=10, L=20$&$K=5, M=10, L=20$&$K=10, M=10, L=20$\\ \hline
\text{Triggered module subset if} $\delta=4.5$& $\{8\}$& $\{1\}$&$\{8\}$&$\{1\}$\\ \hline
\text{Triggered  module subset if } $\delta=5$& $\{1,3,7,8,9\}$& $\{ 1,2,4,8\}$&$\{1,3,7,8,9\}$&$\{1, 2, 4,8\}$\\\hline
\text{Triggered module subset if } $\delta=5.5$ &$\{1,2,3,4,6,7,8,9,10\}$&$ \{ 1,2,3, 4, 8, 9, 10\}$&$\{1,2,3,4,6,7,8,9,10\}$&$\{1,2,3, 4, 8, 9, 10\}$\\ \hline
\text{Triggered module subset if } $\delta=6$ &$\{1,2,3,4,5,6,7,8,9,10\}$&$\{1,2, 3,4,7,8,9,10\}$&$\{1,2,3,4,5,6,7,8,9,10\}$&$\{1, 2, 3, 4, 7,8,9,10\}$\\ \hline
\end{tabular}
}
\label{tab:2}
\end{table*}

We first show the convergence behavior of the two-block ADMM algorithm. Table \ref{tab:2} summarizes the triggered module subsets by using the proposed ADMM algorithm and by the centralized algorithm CVX.
In Table \ref{tab:2}, we consider two different network settings, including $K=5, M=10$ and $K=10, M=10.$ For the two settings, we  present the triggered module subsets obtained by  ADMM algorithm and CVX with different values of sparse parameter $\delta,$ respectively.
It is clear from Table \ref{tab:2} that the ADMM algorithm can converge to the CVX solution.
\begin{figure*}[!t]
	\centering
	\subfigure{
		\begin{minipage}[t]{0.38\linewidth}
			\includegraphics[width=2.5in]{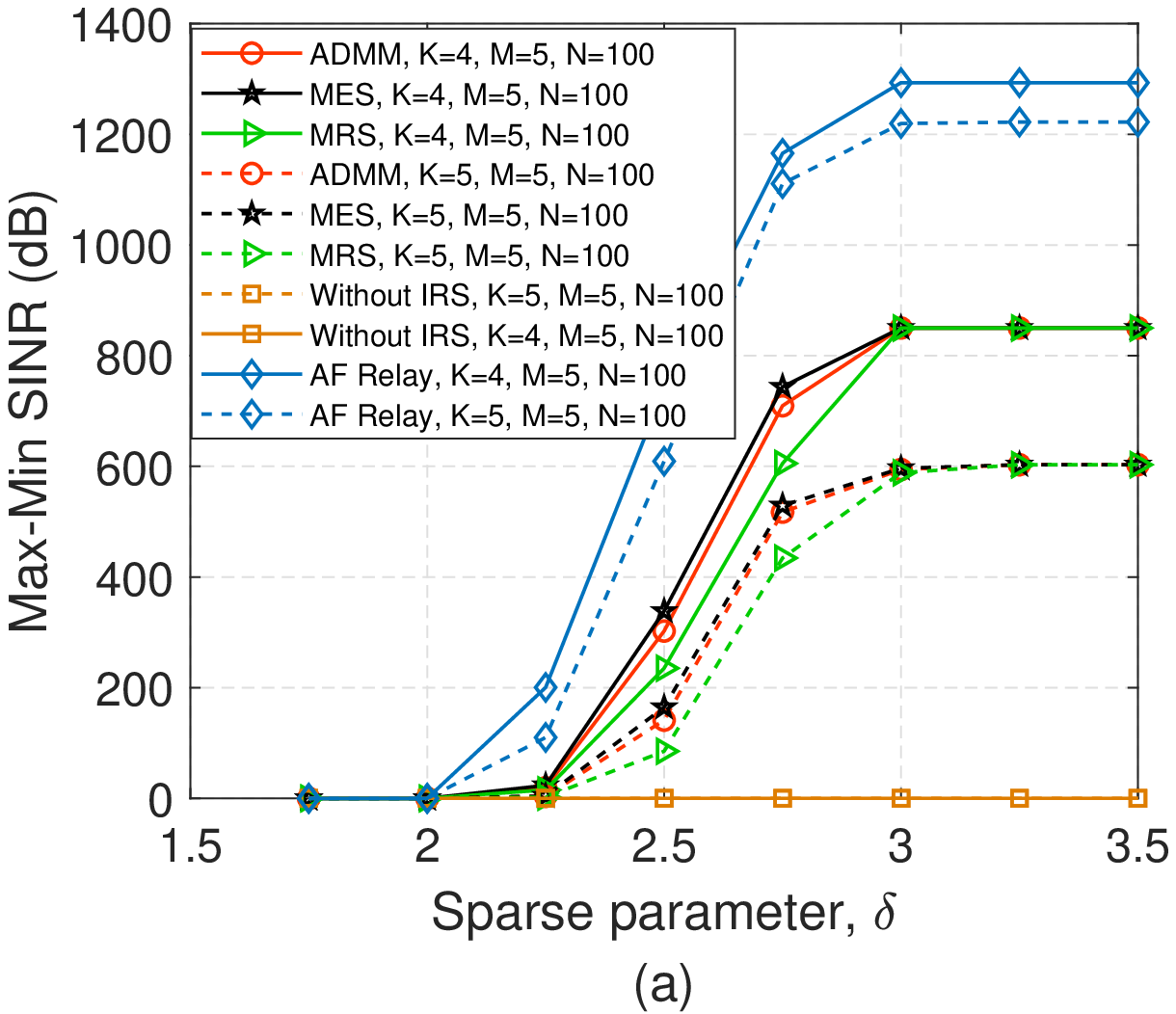}
		\end{minipage}
		\label{fig:3-1}
	}
	\centering
	\subfigure{
		\begin{minipage}[t]{0.38\linewidth}
			\includegraphics[width=2.5in]{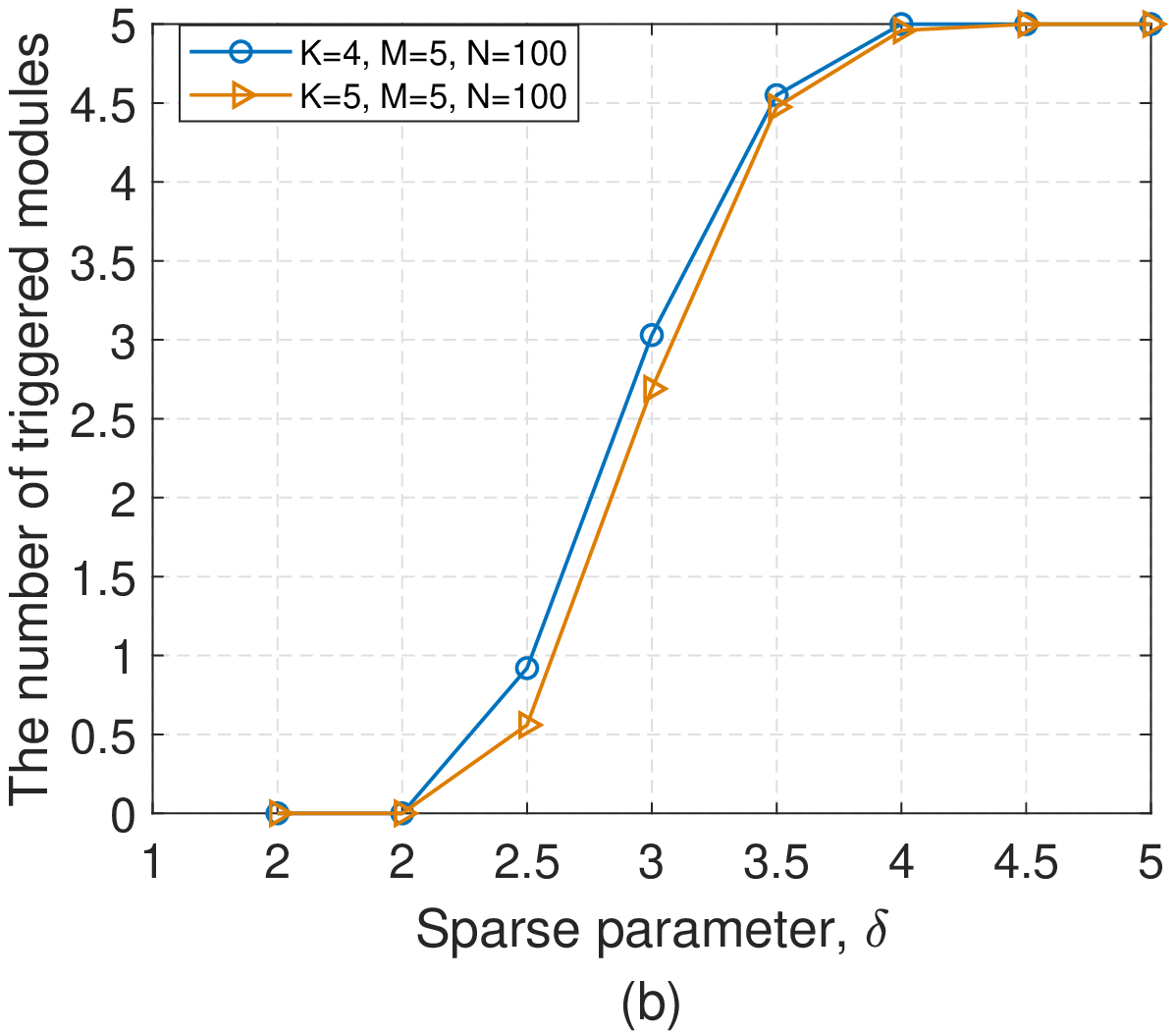}
		\end{minipage}
		\label{fig:3-2}
	}
	\caption{ (a) Max-min SINR and (b) the number of triggered modules as functions of sparse parameter $\delta$ with $K\in\{4, 5\}$, for $M=5, p^{\max}=20\text{~dBm}.$}
	\label{fig:3}
\end{figure*}
\subsection{Performance Comparison}
{For IRS-aided communication systems, we evaluate the performance of the proposed two-block ADMM algorithm with two baseline schemes in our simulations.
For baseline 1 (i.e., the case without IRS),  only the S-D direct link is considered where the number of reflecting elements at IRS is set as  $N=0.$
For baseline 2, the classic Amplify-and-Forward (AF) relaying is considered where the maximum transmit power at relay is the same as the total budgets at the ST side \cite{Huang2019Reconfigurable}, i.e., $P_r^{\max}=K p^{\max},$ for ensuring a fairness comparison.
Besides, in order to show the effectiveness of our proposed ADMM algorithm, we consider two conventional methods in our simulations,  denoted as method of exhaustive search (MES) and method of randomly selecting the triggered modules (MRS), respectively.
The premise of performance comparison between the two methods is to use the same sparsity constraint to control the size of triggered modules. In this paper, the number of triggered modules is determined by the ADMM algorithm.
Specifically, both MES and MRS perform their respective triggered modules identification and solve the conventional max-min SINR problem via bisection feasibility checking.
For the AF relay, the number of antennas used at relay is $card(\text{ subset of triggered modules})\cdot L,$ i.e., the number of antennas is the same as the number of reflecting elements.
}
\begin{figure*}[!t]
	\centering
	\subfigure{
		\begin{minipage}[t]{0.38\linewidth}
			\includegraphics[width=2.5in]{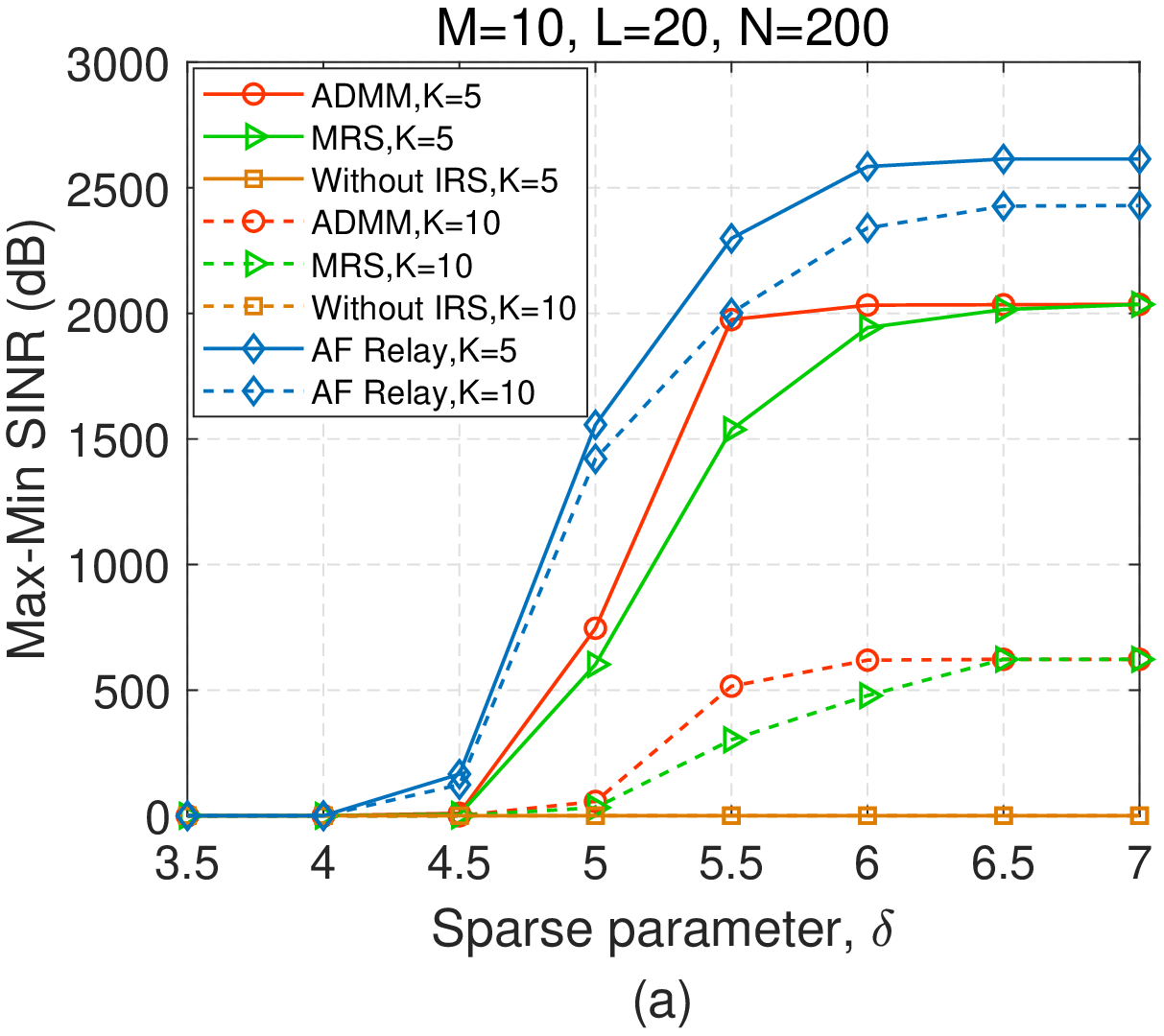}
		\end{minipage}
		\label{fig:4-1}
	}
	\centering
	\subfigure{
		\begin{minipage}[t]{0.38\linewidth}
			\includegraphics[width=2.5in]{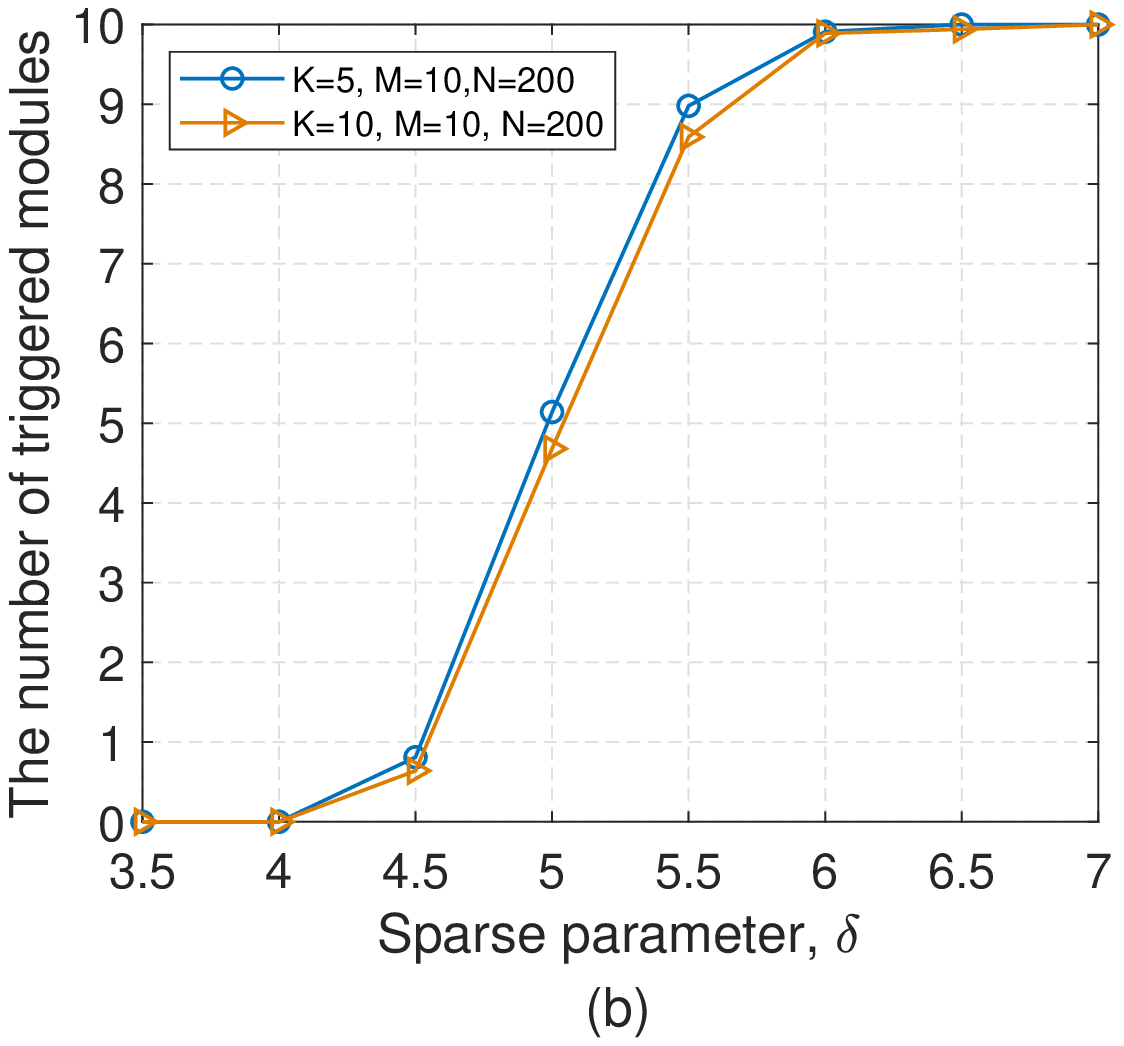}
		\end{minipage}
		\label{fig:4-2}
	}
	\centering
	\subfigure{
		\begin{minipage}[t]{0.38\linewidth}
			\includegraphics[width=2.5in]{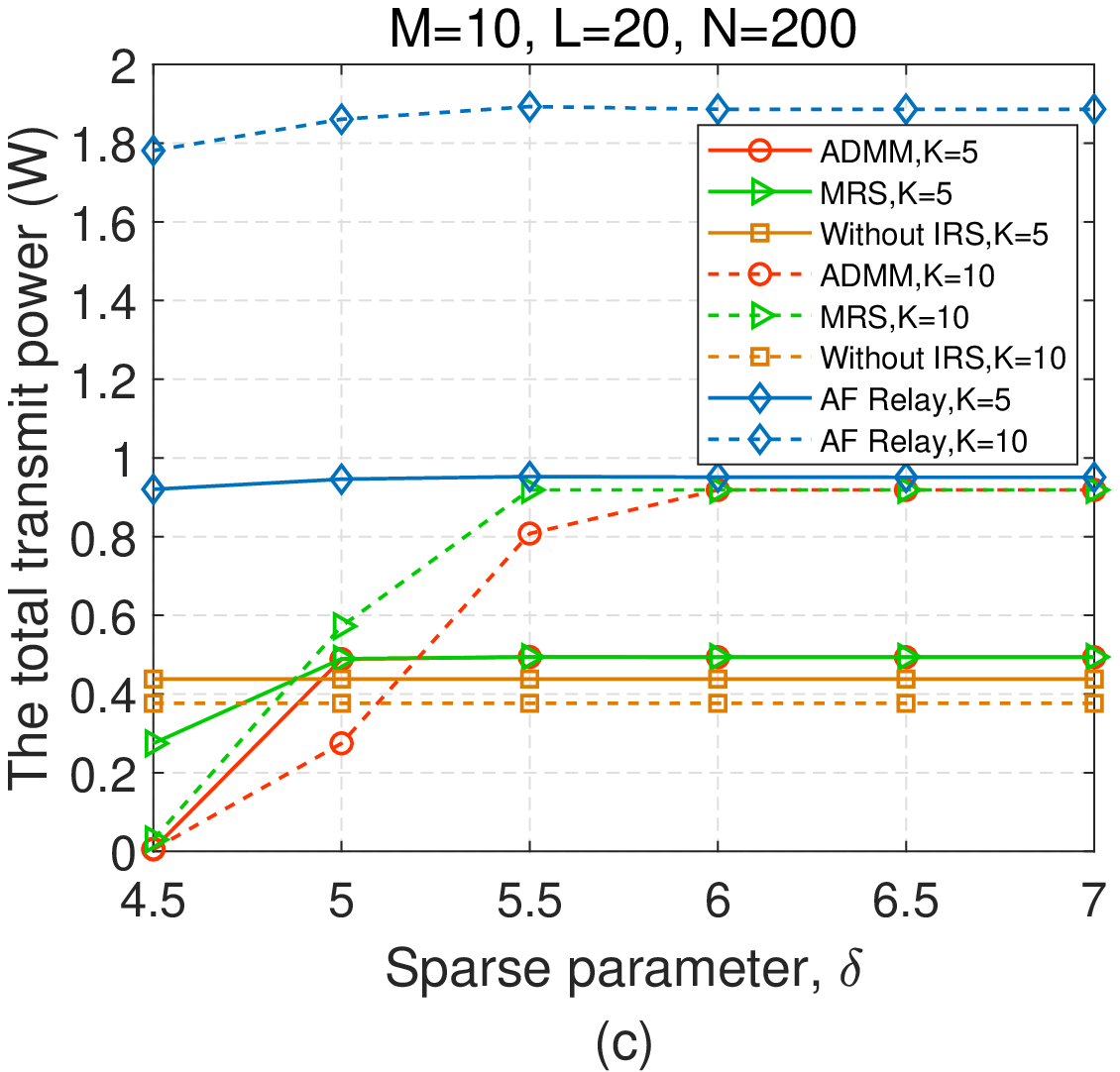}
		\end{minipage}
		\label{fig:4-3}
	}
	\centering
	\subfigure{
		\begin{minipage}[t]{0.38\linewidth}
			\includegraphics[width=2.5in]{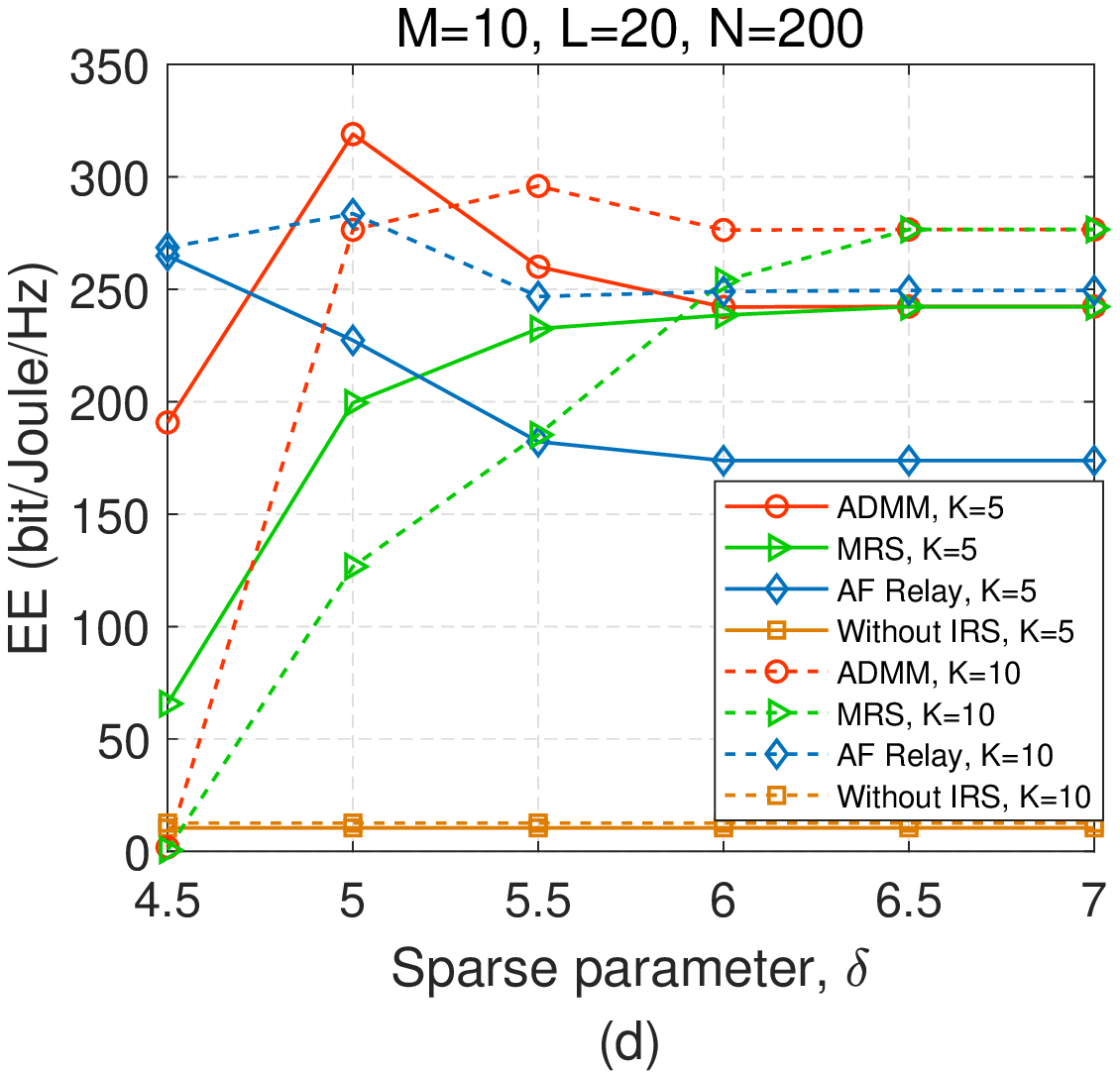}
		\end{minipage}
		\label{fig:4-4}
	}
	\centering
	\caption{(a) Max-min SINR, (b) the number of triggered modules, (c) the total transmit power, and (d) EE versus the sparse parameter $\delta$ using $K=5, 10$, for $M=10$ and $p^{\max}=20\text{~dBm}.$}
	\label{fig:4}
\end{figure*}
\subsubsection{{Baseline 2: AF Relay}}

{We consider the simulation setup for AF relaying, where the conventional AF relay equipped with $N$ antennas in the place of the IRS structure. Besides, we consider the same user terminals' positions and channel realizations in both IRS and AF relay cases.
Similar to \cite{Yang2019Low}, we consider coordinated relay beamforming design to reduce inter-pair interference. The received signal vectors ${\mathbf r}\in{\cal C}^{N\times 1}$ and ${\mathbf y}\in{\cal C}^{K\times 1}$ at relay and all DTs, respectively, are given by
\begin{equation}\label{add:9}
\begin{aligned}
{\mathbf r}=&{\mathbf H}{\mathbf P}{\mathbf z}+{\mathbf v},\\
{\mathbf y}=&{\mathbf G}{\mathbf V}{\mathbf H}{\mathbf P}{\mathbf z}+{\mathbf G}{\mathbf W}{\mathbf v}+{\mathbf u},
\end{aligned}
\end{equation}
where ${\mathbf z}=[z_1, \ldots, z_K]^T\in{\cal C}^{K\times 1},$ ${\mathbf u}=[u_1, \ldots, u_K]^T\in{\cal C}^{K\times 1}$, ${\mathbf H}=[{\mathbf h}_1, \ldots, {\mathbf h}_K]\in{\cal C}^{N\times K},$ ${\mathbf G}=[{\mathbf g}_1, \ldots, {\mathbf g}_K]\in {\cal C}^{K\times N},$ ${\mathbf P}=\text{diag}[\sqrt{p_1}, \ldots, \sqrt{p_K}]\in{\cal C}^{K\times K},$  ${\mathbf v}\in {\cal C}^{N\times 1}$ is the additional white Gaussian noise (AWGN) vector at the receiver of relay with covariance matrix $\sigma_v^2{\mathbf I},$ and ${\mathbf V}\in{\cal C}^{N\times N}$ is the beam matrix at the relay.  In detail, we construct ${\mathbf V}$ by zero forcing (ZF) beam matrix to cancel the interference from/to other ST-DT pairs. Let ${\mathbf v}_{T, k}^{\text{ZF}}\in{\cal C}^{N\times 1}$ denote the transmit beam vector for $d_k$ and ${\mathbf v}_{R,k}^{\text{ZF}}\in{\cal C}^{N\times 1}$ denote the receive beamformer for  $s_k$.
The ZF beam matrix can be constructed  by transmit beam vector ${\mathbf V}_T^{\text{ZF}}=[{\mathbf v}_{T,1}^{\text{ZF}}, \ldots, {\mathbf v}_{T, K}^{\text{ZF}}]\in{\cal C}^{N\times K}$ and the receive beamformer ${\mathbf V}_R^{\text{ZF}}=[{\mathbf v}_{R, 1}^{\text{ZF}}, \ldots, {\mathbf v}_{R, K}^{\text{ZF}}]\in{\cal C}^{N\times K}$, i.e.,  ${\mathbf V}={\mathbf V}_T^{\text{ZF}}{\mathbf V}_R^{\text{ZF}\dag}.$
Subsequently, according to the ZF criterion, the receive SINR at $d_k$ can be expressed as
\begin{equation}\label{add:10}
\begin{aligned}
\text{SINR}_k^{\text{AF}}=\frac{p_k(\zeta_k^{\text{AF}})^2}
{(\zeta_k^{\text{AF}})^2\sigma_v^2||{\mathbf e}_k^T\left( {\mathbf H}^{\dag}{\mathbf H}\right)^{-1}{\mathbf H}^{\dag}||_2^2+\sigma^2},
\end{aligned}
\end{equation}
where $\zeta_k^{\text{AF}}$ in (\ref{add:11}) is a scalar which means the amplify factor for ST-DT $k,$ which is written as 
\begin{equation}\label{add:11}
\zeta_k^{\text{AF}}=\left(\frac{P_r^{\max}(p_k+\sigma_v^2[ ( {\mathbf H}^{\dag}{\mathbf H}   )^{-1}  ]_{k,k})[( {\mathbf G}{\mathbf G}^{\dag})^{-1} ]_{k,k}}
{\sum_{k=1}^K(p_k+\sigma_v^2[( {\mathbf H}^{\dag}{\mathbf H})^{-1}]_{k,k} )^2  ( [({\mathbf G}{\mathbf G}^{\dag}  )^{-1}  ]_{k,k}  )^2}  \right)^{\frac{1}{2}}.
\end{equation}
We refer the interested readers to \cite{Yang2019Low} for a detailed proof for the above result.
}

{Then, for the case of AF relay, we consider the following max-min SINR problem for the design of ${\mathbf p}=[p_1, \ldots, p_K]^T$:
\begin{equation}\label{add:12}
\begin{aligned}
&\max_{\{p_k\}_{k=1}^K}\min_k~\text{SINR}_k^{\text{AF}}\\
\text{s.t.}~&p_k\leq p_k^{\max}, \forall k\in {\cal K}.
\end{aligned}
\end{equation}
We can solve (\ref{add:12}) efficiently using bisection search along with a convex feasibility problem. }

\subsubsection{{Two-block ADMM Algorithm Effectiveness  Verification}}
{To gain insight into the two-block ADMM algorithm, detailed performance comparison is  provided firstly for a small number of modules, which allows comparing with the method of exhaustive search (MES) solution.
Nevertheless, for MES, an exhaustive combinatorial search is required over all possible cases of $2^{M}$. As a result, we can simply set $M=5$ for implementing the brute force search to obtain the optimal solution and provide the performance upper bound for the proposed ADMM algorithm.
Figure \ref{fig:3} shows the influence of the sparse parameter $\delta$ on the max-min SINR and the number of triggered modules,  using different number of user pairs $K\in \{4, 5\}.$
As expected, the number of triggered modules is effectively controlled by $\delta$ and the sparsity of solution is sensitive to the value of $\delta.$
As seen in Fig. \ref{fig:3-1},  the max-min SINR of both cases (i.e., $K\in\{4, 5\}$) achieved by all the schemes first increases and then respectively approaches a constant value with the increasing value of sparse parameter $\delta.$
Correspondingly, in Fig. \ref{fig:3-2}, the number of triggered modules increases monotonically until reaching the total number of available modules.
In fact, introducing larger $\delta$ makes the module size constraint more relaxed, which leads to a less sparse solution.
In other words, when $\delta$ increases, more modules (or reflecting elements) are available for cooperative communication, both SINR and the number of triggered modules increase accordingly.
In addition, there is a special point,  $4,$  which can be obtained by
$\frac{-0.01+\sqrt{(0.01)^2+\sqrt{16\times5\times 5\times 100\times 0.1}}}
{2}\approx 4$ and   $\frac{-0.01+\sqrt{(0.01)^2+\sqrt{16\times5\times 4\times 100\times 0.1}}}
{2}\approx 4.$
As indicated in Lemma 1, the triggered module size constraint is inactive when $\delta>4,$ consequently, the number of triggered modules is $5.$
As seen from Fig. \ref{fig:3-1}, the AF relay outperforms the other four methods, this is mainly because apart from reflecting structure, an active terminal is employed at AF relay.
From the results we notice furthermore that, the considered ADMM algorithm achieves the performance which is slightly close to that of MES and outperforms two other methods besides AF relaying.
}
\begin{figure*}[!t]
\centering
\subfigure{
\begin{minipage}[t]{0.38\linewidth}
\includegraphics[width=2.5in]{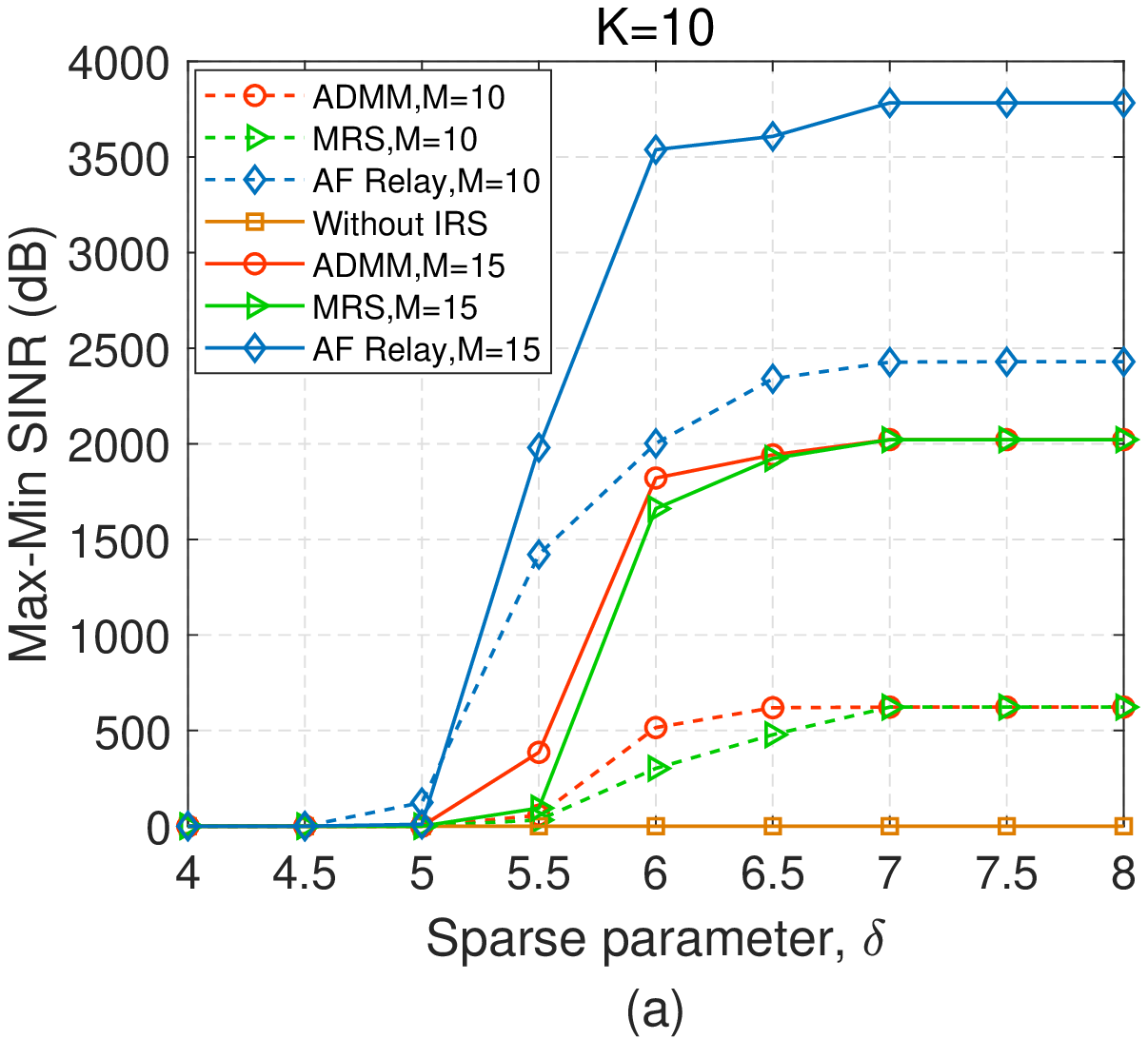}
\end{minipage}
\label{fig:5-1}
}
\centering
\subfigure{
\begin{minipage}[t]{0.38\linewidth}
\includegraphics[width=2.5in]{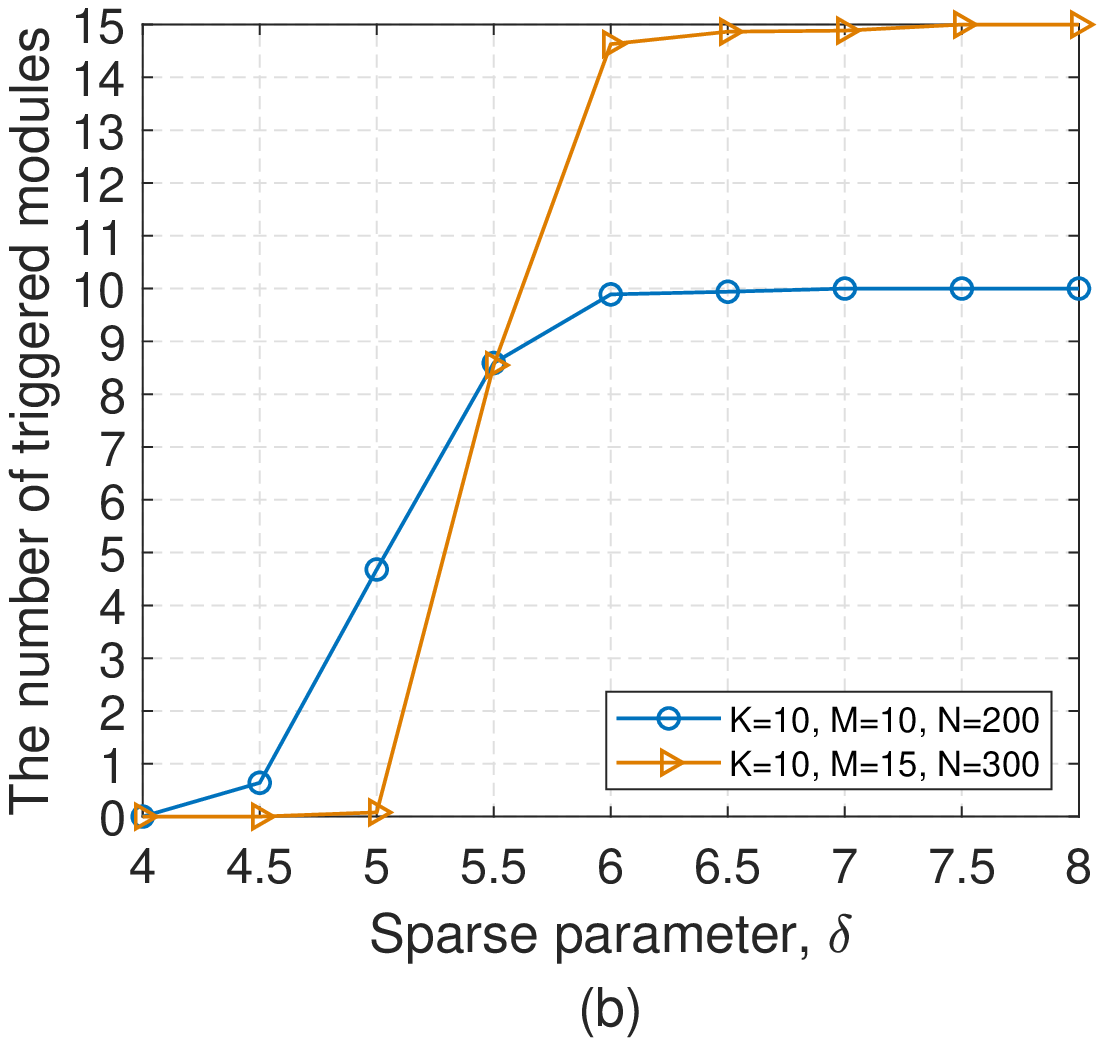}
\end{minipage}
\label{fig:5-2}
}
\centering
\subfigure{
\begin{minipage}[t]{0.37\linewidth}
\includegraphics[width=2.5in]{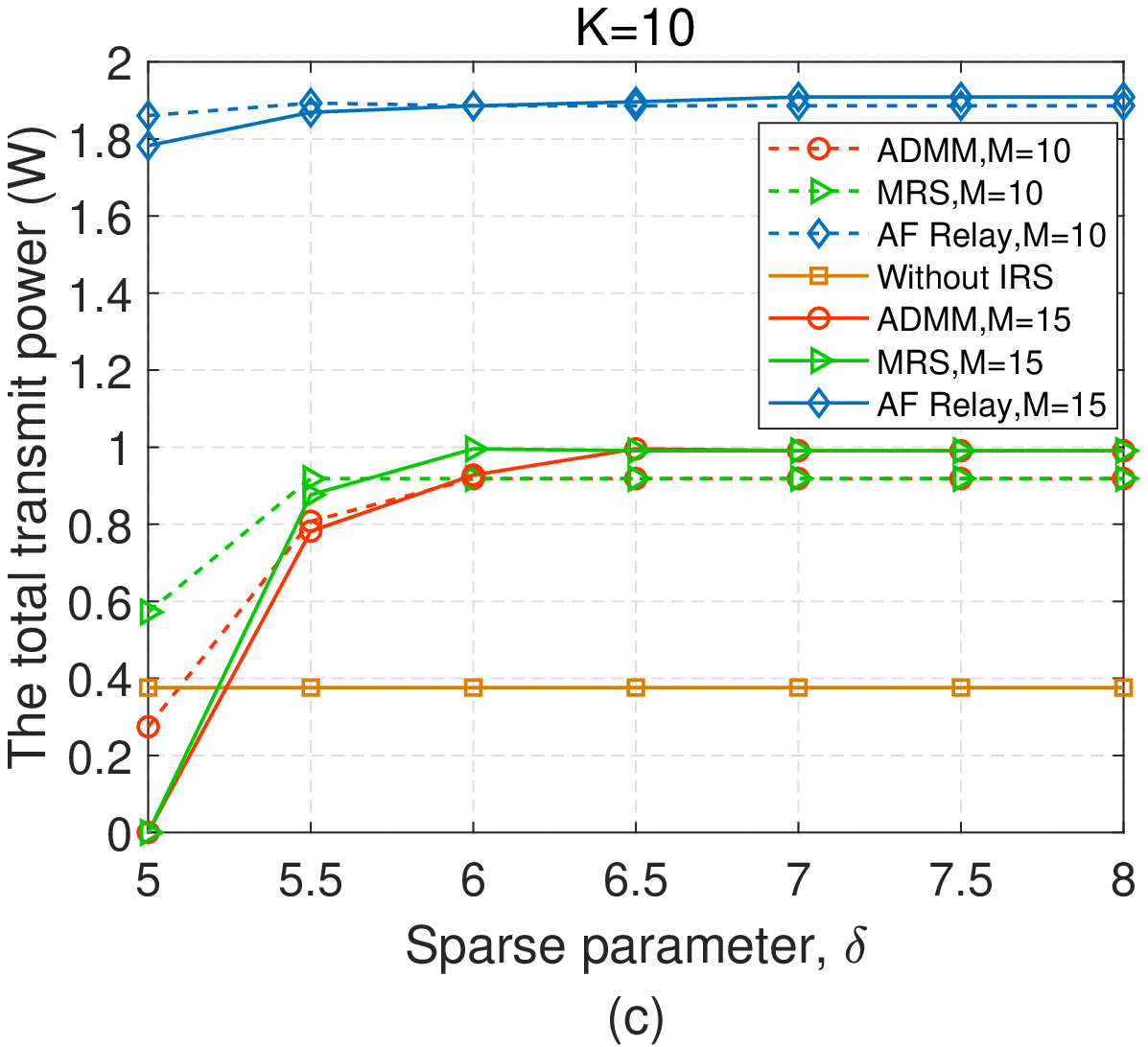}
\end{minipage}
	\label{fig:5-3}
}
\centering
\subfigure{
\begin{minipage}[t]{0.37\linewidth}
\includegraphics[width=2.5in]{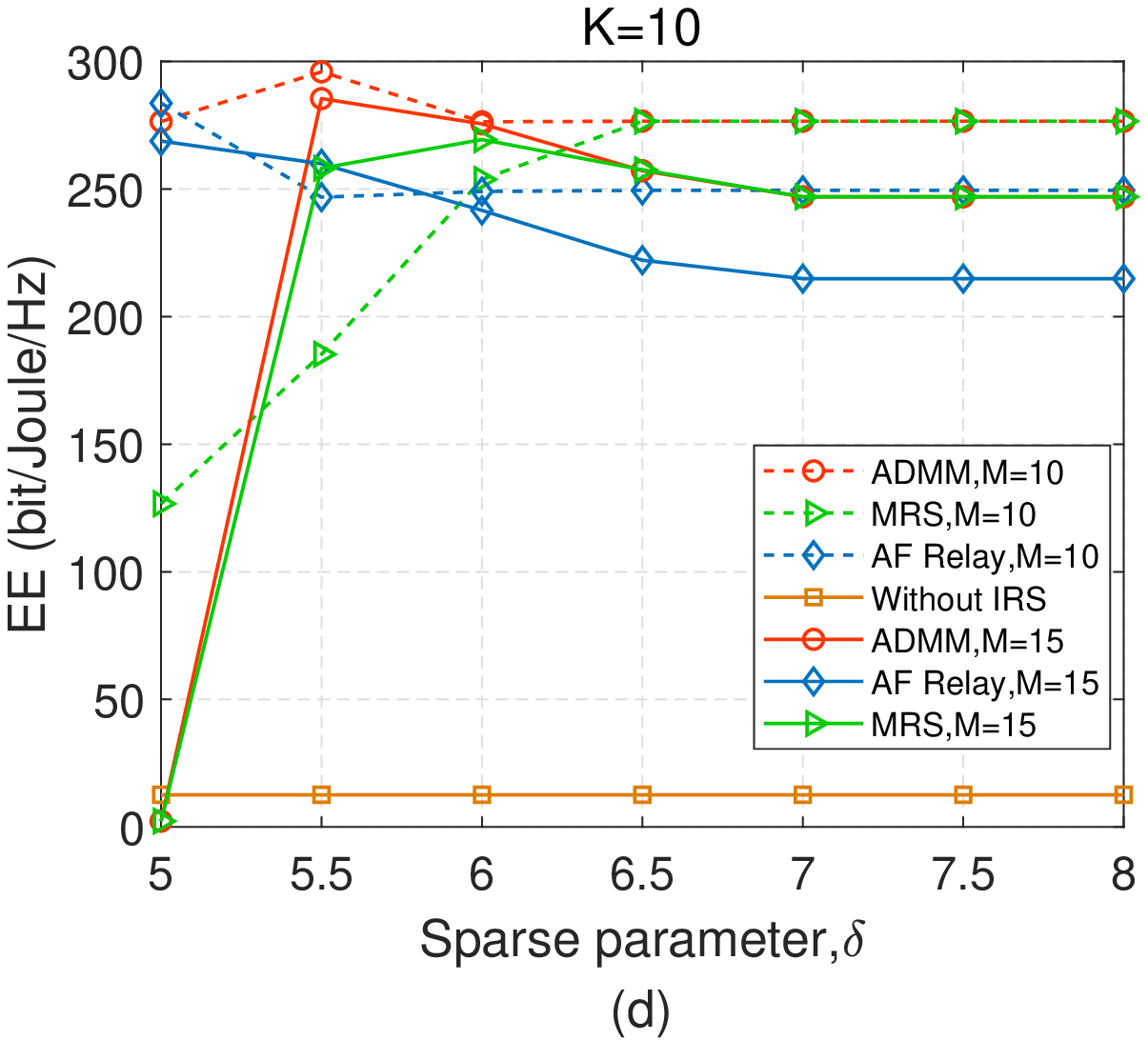}
\end{minipage}
\label{fig:5-4}
}
\centering
\caption{(a) Max-min SINR, (b) the number of triggered modules, (c) the total transmit power, and (d) EE versus the sparse parameter $\delta$ using $M=10, 15$, for $k=10$ and $p^{\max}=20\text{~dBm}.$}
\label{fig:5}
\end{figure*}
\subsubsection{{Performance Comparison under different $K$
}}
{Figures \ref{fig:4-1} and \ref{fig:4-2} illustrate the effects of the number of S-D pairs $K$ on the SINR performance and the number of triggered modules, respectively, based on all above schemes besides the MES.
Two simulation cases with $K=5$ and $K=10$ are shown with the same number of modules $M=10$ at the IRS and the maximum transmit power of each ST is $p^{\max}=20\text{~dBm}.$
From the results, we observe that the SINR achieved by all above schemes besides the baseline 1 first increases and then remain constant, when $\delta$ increases.
Apparently, the reduction of SINR for all above schemes besides the baseline 1 due to increasing $K$ is significant.
This in essence attributes to that more interference will be induced from concurrent transmissions if the IRS serves for more ST-DT pairs.
As seen in Fig.\ref{fig:4-2}, for given values of $M$ and $p^{\max}$, a sparser solution will be achieved for $K=10$ than that of $K=5.$
This is due to the fact that the module size constraint is more stringent with $K=5$ than that of $K=10$ for the same value of $\delta,$ based on the result of Lemma 1.
In addition to our observations in Fig. \ref{fig:4-1} and \ref{fig:4-2} with respect to SINR and the number of triggered modules, in Fig. \ref{fig:4-3}, we evaluate and compare the total transmit power versus the sparse parameter using different number of user pairs for a given value of $M.$
As expected,  for $K\in\{5, 10\},$  the total transmit power by the ADMM algorithm is lower than that of MRS for a given value of $\delta.$
}

\begin{figure*}[!t]
	\centering
	\subfigure{
		\begin{minipage}[t]{0.28\linewidth}
			\includegraphics[width=1.8in]{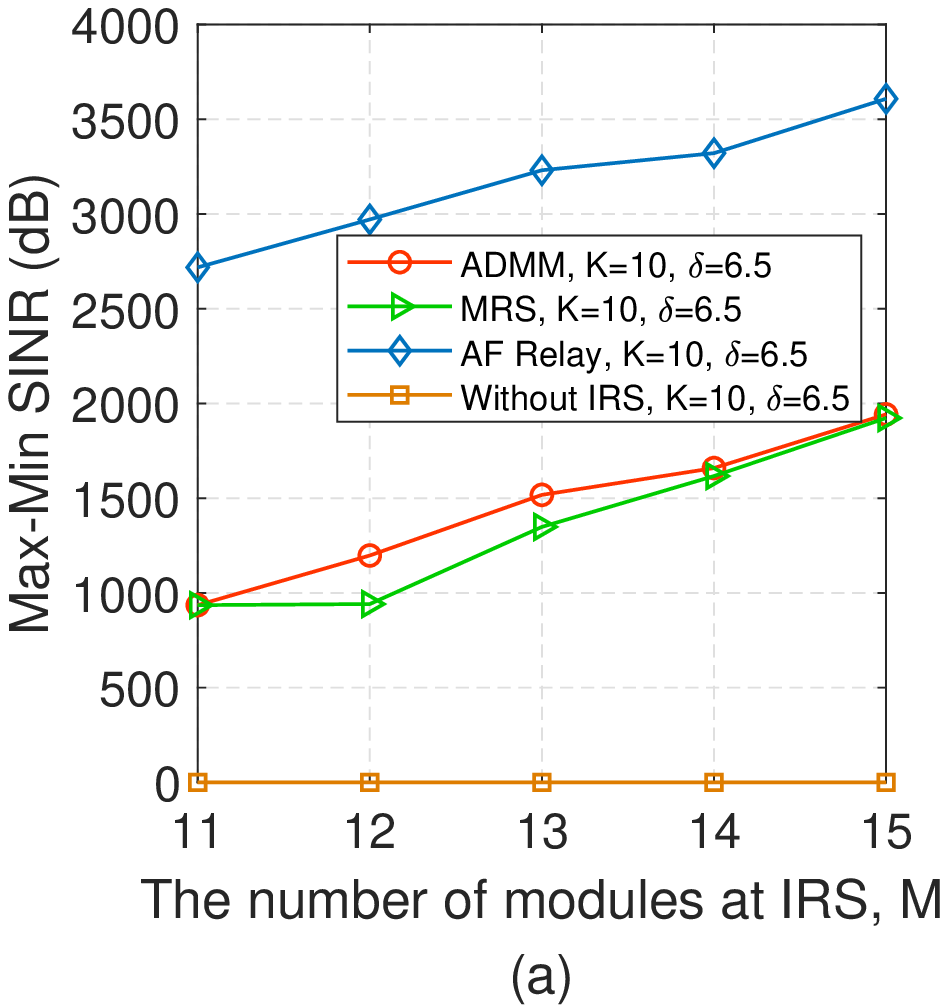}
			\label{fig:6-1}
		\end{minipage}
	}
	\subfigure{
		\begin{minipage}[t]{0.28\linewidth}
			\includegraphics[width=1.8in]{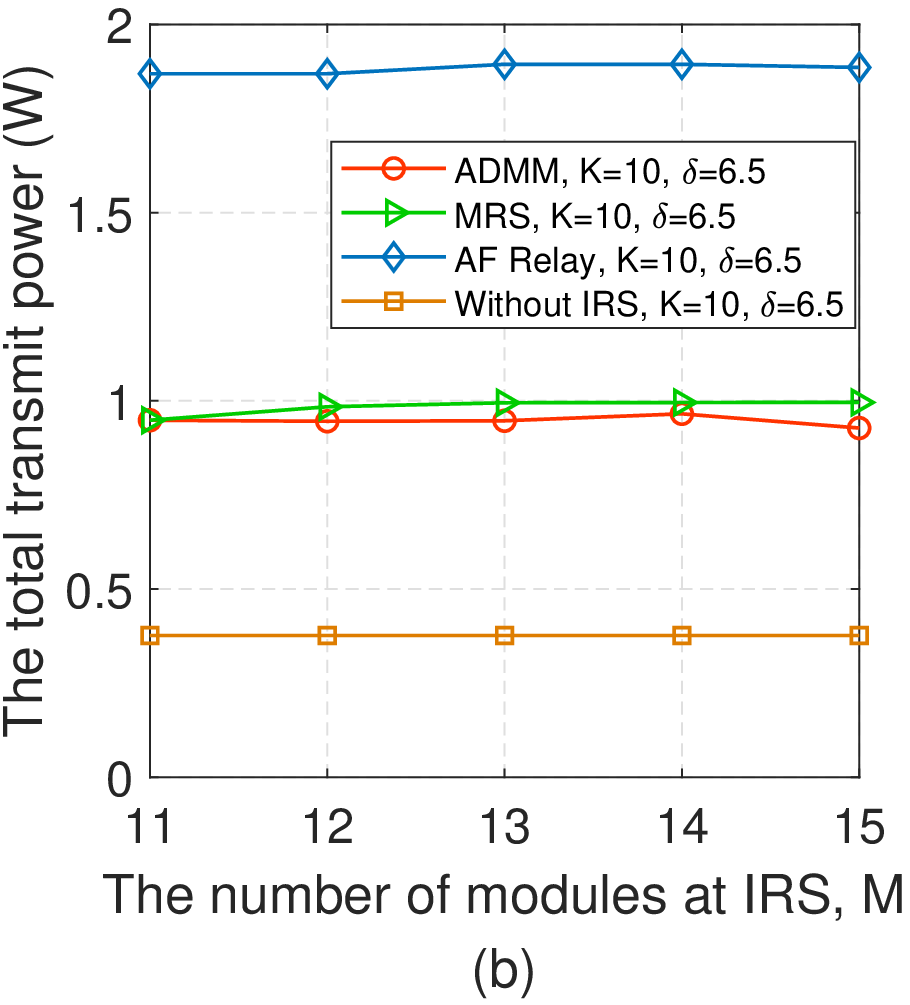}
			\label{fig:6-2}
		\end{minipage}
	}
	\subfigure{
		\begin{minipage}[t]{0.28\linewidth}
			\includegraphics[width=1.8in]{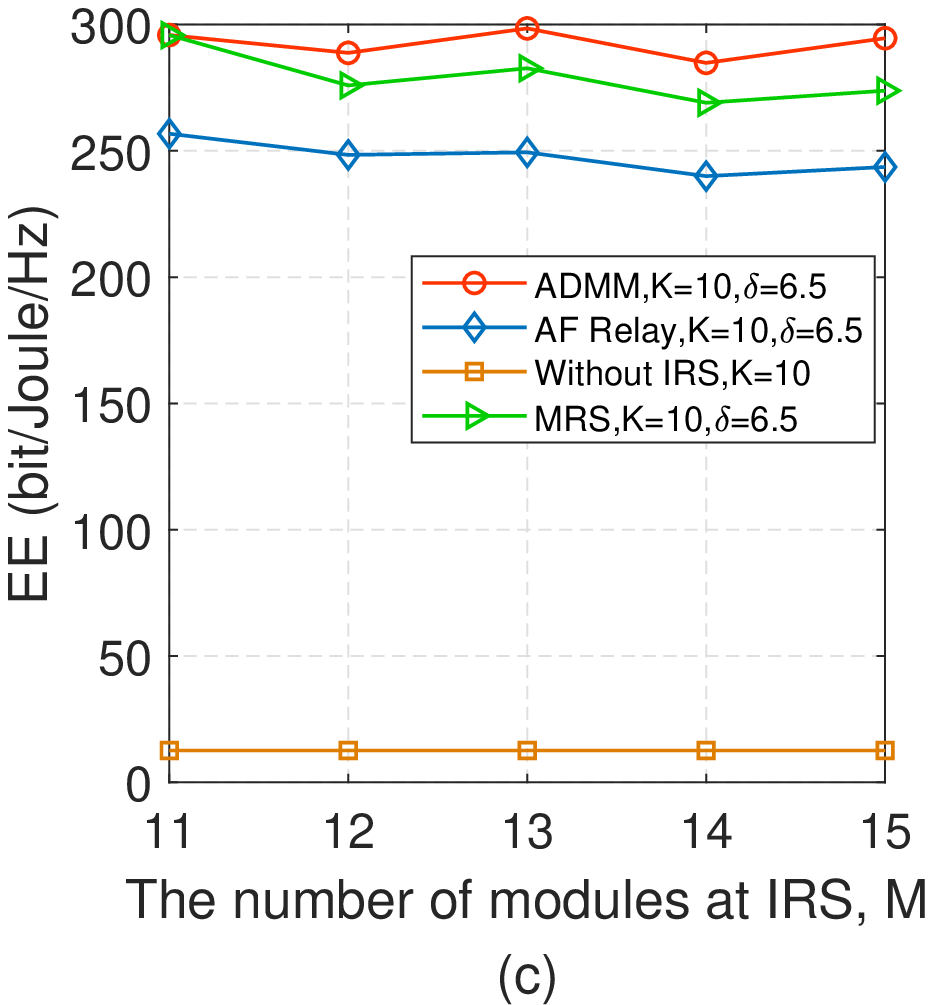}
			\label{fig:6-3}
		\end{minipage}
	}
	\centering
	\caption{(a) Max-min SINR, (b) the total transmit power, and (c) EE versus the number of modules $M$, for $K=10, p^{\max}=20\text{~dBm},$ and $\delta=6.5.$ }
	\label{fig:6}
\end{figure*}
\begin{figure*}[!t]
	\centering
	\subfigure{
		\begin{minipage}[t]{0.28\linewidth}
			\includegraphics[width=1.8in]{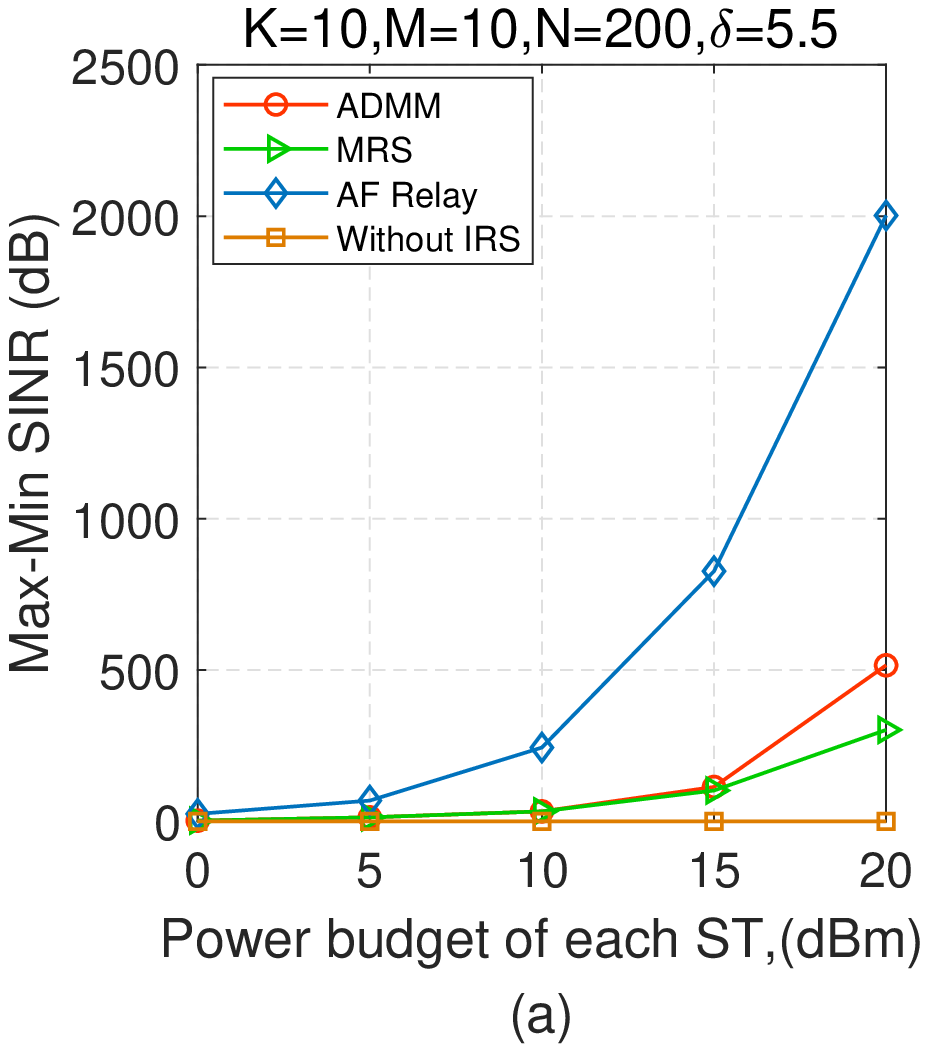}
			\label{fig:7-1}
		\end{minipage}
	}
	\subfigure{
		\begin{minipage}[t]{0.28\linewidth}
			\includegraphics[width=1.8in]{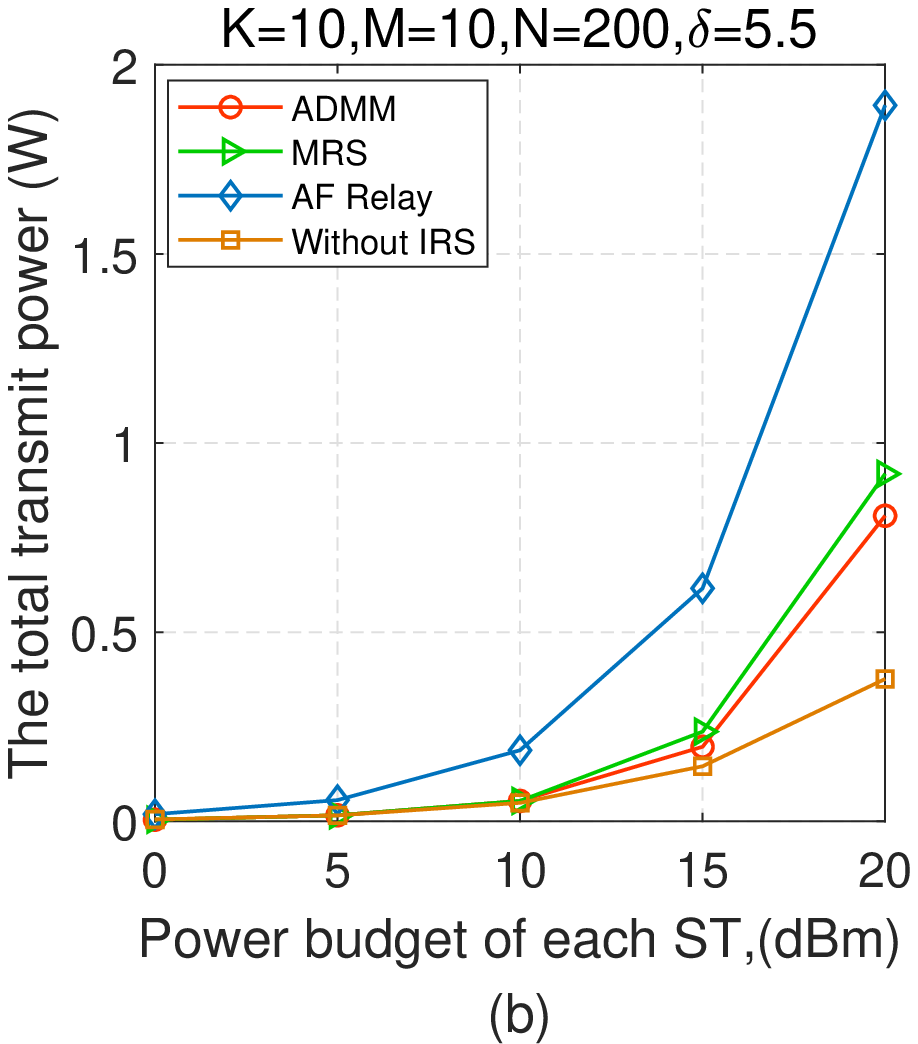}
			\label{fig:7-2}
		\end{minipage}
	}
	\subfigure{
		\begin{minipage}[t]{0.28\linewidth}
			\includegraphics[width=1.8in]{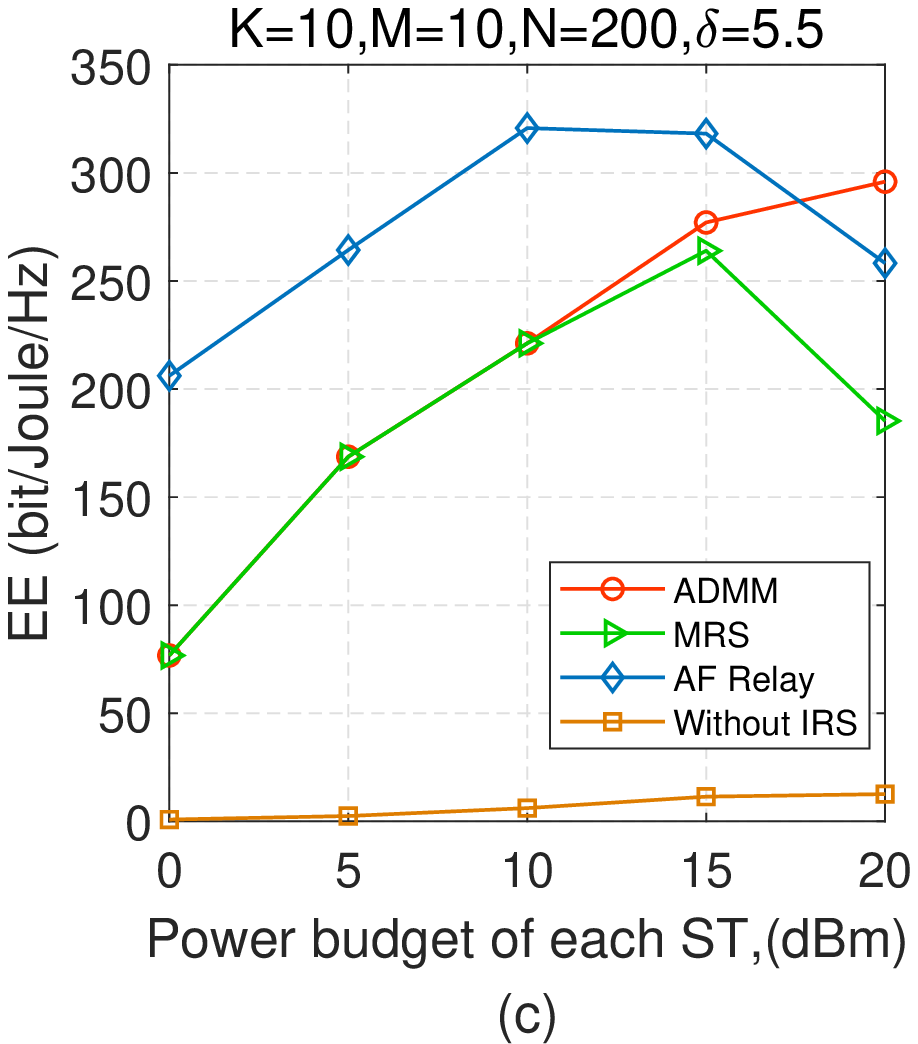}
			\label{fig:7-3}
		\end{minipage}
	}
	\centering
	\caption{(a) Max-min SINR, (b) the total transmit power, and (c) EE versus the maximum transmit power budget of each ST $p^{\max}$, for $K=10, M=10,$ and $\delta=5.5.$}
	\label{fig:7}
\end{figure*}

{To draw more insight for the superiority of the IRS modular structure, we further compare the EE performance, where the EE (bit/Joule/Hz)  is defined as the ratio of the network achievable sum rate and the overall power consumption, i.e., ${\texttt{EE}}=\frac{\sum_{k=1}^K R_k}{P_{\text{total}}}.$
Inspired by \cite{Huang2019Reconfigurable},  the overall power consumption of IRS-aided system can be expressed as
\begin{equation}
\begin{aligned}
P_{\text{total}}^{\text{IRS}}=&\xi_{\text{ST}}\sum_{k=1}^K p_k+KP_{\text{ST}}+KP_{\text{DT}}\\
&+card(\text{subset of triggered modules})\cdot P(L),
\end{aligned}
\end{equation}
where $P_{\text{ST}}$ and $P_{\text{DT}}$ denote the hardware static power dissipated by each ST and DT, respectively, $\xi_{\text{ST}}$ the circuit dissipated power coefficient at each ST, and $P(L)$ is the power consumption of each module having $L$ reflecting elements.
Correspondingly, for AF relay, the total power consumption is given by
\begin{equation}
\begin{aligned}
P_{\text{total}}^{\text{AF}}=&\xi_{\text{AF}}P_r+\xi_{\text{ST}}\sum_{k=1}^Kp_k+KP_{\text{ST}}+KP_{\text{DT}}\\
&+card(\text{ subset of triggered modules})\cdot L\cdot P_\text{antenna},
\end{aligned}
\end{equation}
where $P_r$ is the total transmit power of relay, $P_{\text{antenna}}=P(L)/L$ represents the power consumption of each antenna, and $\xi_{\text{AF}}$ depends on the efficiency of the relay power amplifier. In addition, all the symbols used in the EE simulations are listed in Table I.
To measure the benefits of the proposed modular triggered mechanism with respect to the existing full activation setting, Fig. \ref{fig:4-4} depicts the EE achieved by all mentioned schemes versus sparse parameter for $M=10.$
As seen in Fig. \ref{fig:4-4}, for $\delta<6$ with $K=5$ ($\delta<6.5$ with $K=10$), the proposed ADMM algorithm in IRS-aided communication significantly outperforms both MRS and  AF relaying.
It is interesting to notice that, for simulation settings $K=5$ and $K=10$, the EE achieved by ADMM algorithm first increases and then decreases until to a saturation value, when the value of $\delta$ increases.
The reason is that when $\delta$ relatively small, e.g., $\delta\in[4.5, 5]$ for $K=5$ ($\delta\in[4.5, 5.5]$ for $K=10$), the increase of SINR dominates the maximizing the EE of system in this regime.
By contrast, as the value of sparse parameter becomes is larger than the optimal $\delta,$ e.g., $\delta>5$ for $K=5$ ($\delta>5.5$ for $K=10$), more and more modules are triggered for cooperative communication, consequently, the circuit power consumption dominates the total power consumption rather than the transmit power consumption.
Therefore, for any given network setting, there is an optimal choice of $\delta,$ which leads to the cost-effective reflecting element schedule.
}

\subsubsection{{ Performance Comparison Under Different $M$ }}
{Figure \ref{fig:5} compares the performance achieved by all the mentioned methods using $K=10,$ for two simulation settings as $M=10$ and $M=15$, varied $\delta$ from $4$ to $8$.
Fig. \ref{fig:5-1} demonstrates the SINR performance against the sparse parameter $\delta$ with the number of ST-DT pairs $K=10.$
We can observe again that the SINR performance order is ``${\text{AF Relay}>\text{ADMM}>\text{MRS}>\text{Without IRS}}$'' for $M\in\{10, 15\}.$
Notably, for $\delta>5.5$, the SINR improvement for all the methods besides baseline 1 by increasing $M$ is significant.
In fact, for the number of modules $M=15$,  more reflecting elements available for cooperative communication if $\delta>5.5,$ as shown in Fig. \ref{fig:5-2}.
It is clear from Fig. \ref{fig:5-2} that the sensitivity of the number of triggered modules to $\delta$ decreases with increasing the value of $\delta$ ($\delta>6$), until constraint (\ref{add:3}) is inactive.  
According to Lemma 1, for $M=15,$ the upper bound of $\delta$ to guarantee the validity of constraint (\ref{add:3}) is $\frac{-0.01+\sqrt{0.01^2+4\sqrt{10\times 15\times 300\times 0.1}}}
{2}\approx 8$, whereas the bound for $M=10$ is about $6.5.$
As a result, for $\delta\in(6.5,8)$, constraint (\ref{add:3}) becomes inactive for $M=10$, while it still works for $M=15.$
Correspondingly, in Fig. \ref{fig:5-3} and \ref{fig:5-4}, we plot the total transmit power and the EE as the function of  $\delta.$
It is clear from Fig. \ref{fig:5-3} and \ref{fig:5-4} that an inappropriate choice of the sparse parameter $\delta$ may lead to severe degradation for the EE.
For $K=10$ there exists an optimal choice of $\delta,$ which maximums the EE.
In addition, for $\delta>5.5$, the EE for all the methods decreases until $\delta$ increases to its upper bound as obtained by Lemma 1.
This is because for $\delta>5.5$, the total power consumption is dominated by the circuit power consumption.
}
\begin{figure*}[!t]
	\centering
	\subfigure{
		\begin{minipage}[t]{0.28\linewidth}
			\includegraphics[width=1.8in]{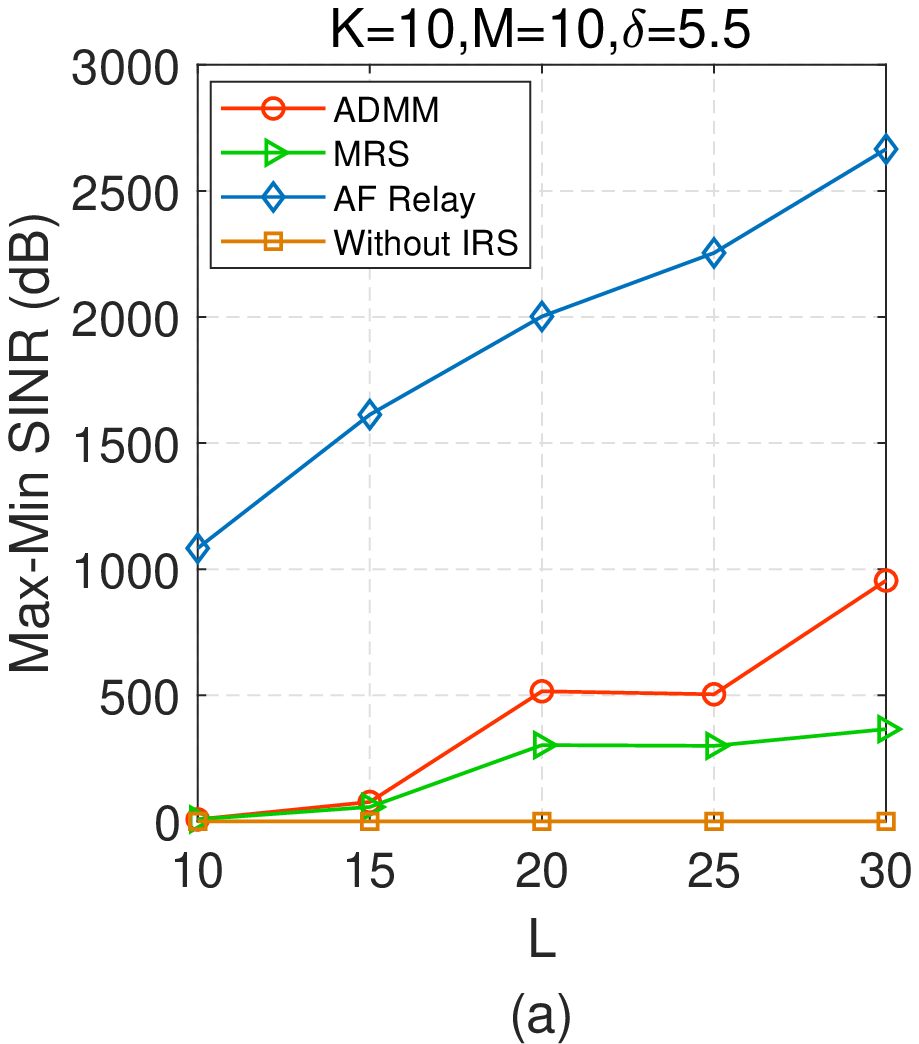}
			\label{fig:8-1}
		\end{minipage}
	}
	\subfigure{
		\begin{minipage}[t]{0.28\linewidth}
			\includegraphics[width=1.8in]{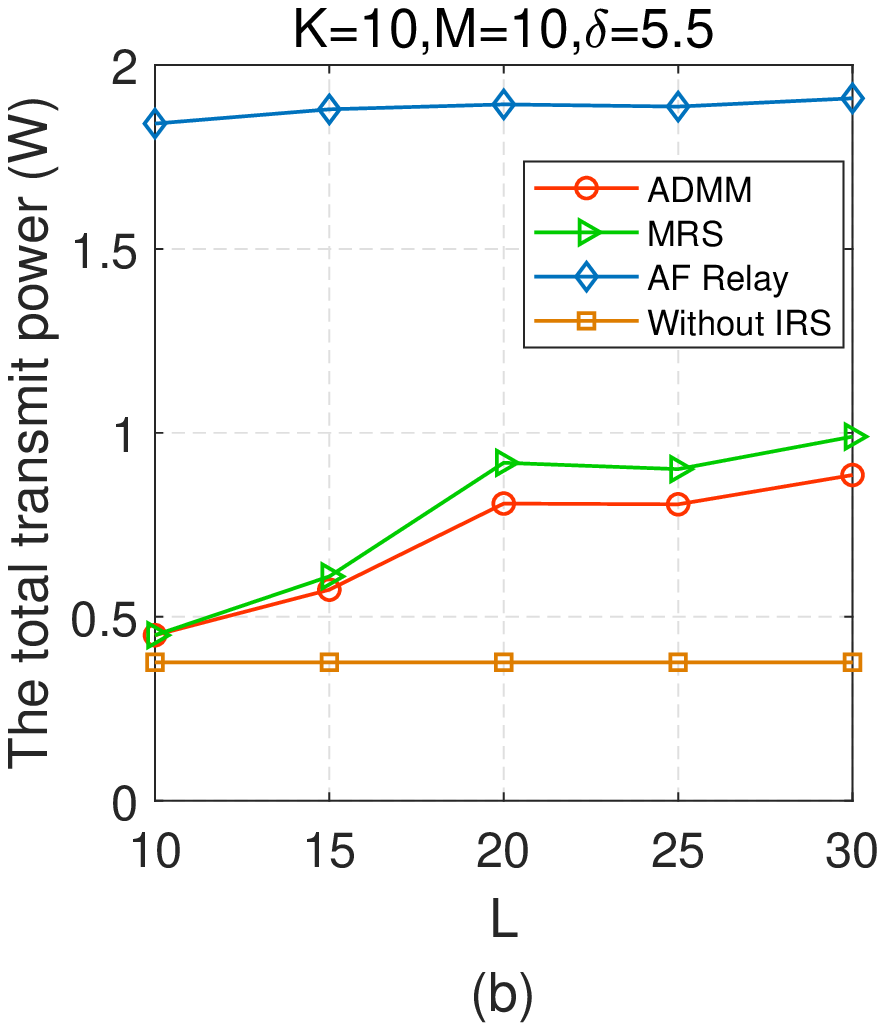}
			\label{fig:8-2}
		\end{minipage}
	}
	\subfigure{
		\begin{minipage}[t]{0.28\linewidth}
			\includegraphics[width=1.8in]{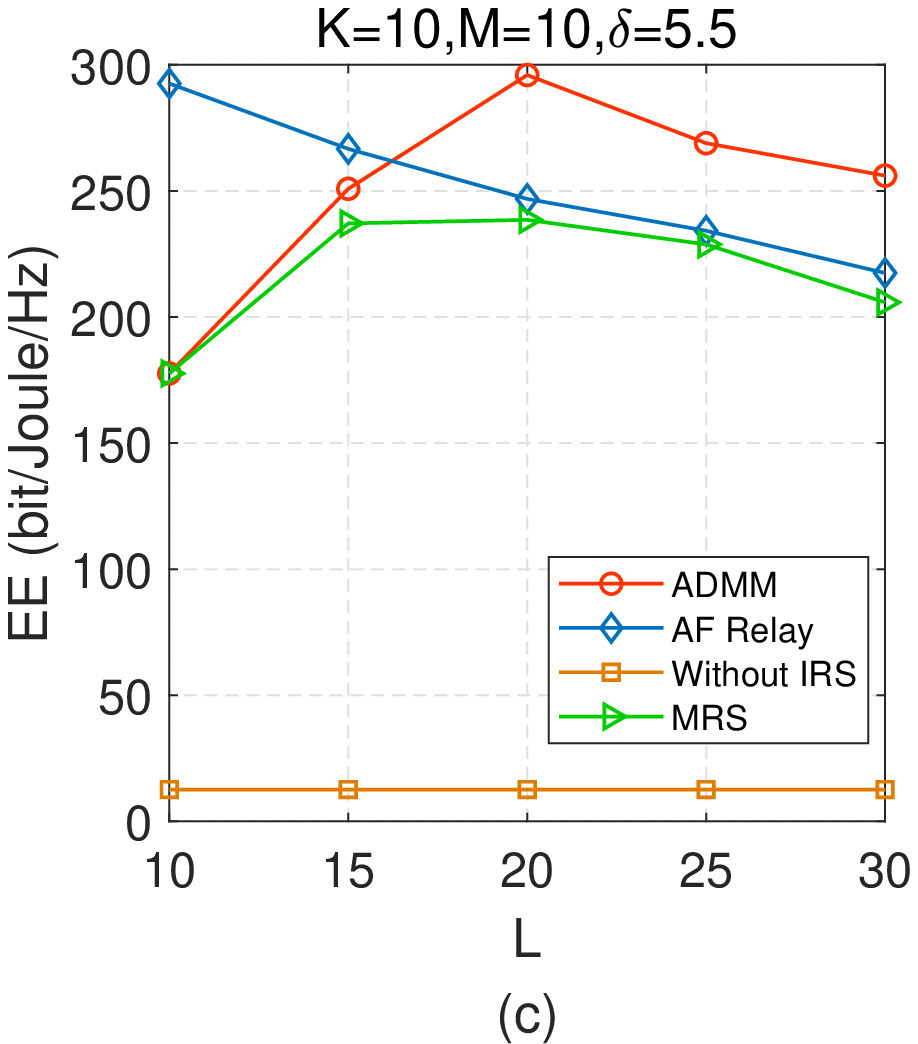}
			\label{fig:8-3}
		\end{minipage}
	}
	\centering
	\caption{(a) Max-min SINR, (b) the total transmit power, and (c) EE versus the number of reflecting elements at each module $L$, for $K=10, M=10, p^{\max}=20\text{~dBm}$, and $\delta=5.5.$}
	\label{fig:8}
\end{figure*}
\subsubsection{{Performance Comparison versus $M$}}
{Figure \ref{fig:6} compares the performance achieved by all the methods using $\delta=6.5$ in the  network setting $K=10$, as a function of the number of modules $M.$
According to Fig. 4(b), for $\delta=6.5,$ the cardinality of the triggered module subset increases very slowly with the number of modules $M,$ varied from $11$ to $15.$
Nevertheless, the module size constraint (12) is still valid for all $M\in\{11, 12, 13, 14, 15\}$.
Figs. \ref{fig:6-1} and \ref{fig:6-2} depict the SINR and the total transmit power consumption, respectively, whereas Fig. \ref{fig:6-3} shows the EE.
From the results, we observe that the SINR achieved by all the methods besides baseline 1 increases as the number of modules, $M$, increases.
As is apparent, the AF relay solution increases the SINR performance at the price of heavy degradation in the EE, irrespectively of the value of $M.$
Comparing Fig. \ref{fig:6-1} to \ref{fig:6-2}, for the given $K=10, \delta=6.5,$ the SINR performance is sensitive to $M$ , when $M\leq 14.$
In other words, for the given values of $K$ and $\delta,$ the improvement of SINR depends only on the triggered modules.
In addition, in Fig. \ref{fig:6-3}, it can be observed that the proposed ADMM outperforms the remaining three algorithms for $K=10$ and $\delta=6.5.$
For given network setting $K=10$ and $\delta=6.5,$  there is an optimal value of $M=13.$
Notably, in Fig. \ref{fig:6-3}, there is a turning point, $14$, due to the less total transmit power at the source side.
}

\subsubsection{{Performance Comparison versus $p^{\max}$}}
{Figure \ref{fig:7} depicts the performance versus the maximum transmit power budget at each ST with $K=10, M=10,$ and $\delta=5.5.$
The observation of Fig.  \ref{fig:7} can be explained with the aid of Lemma 1 and Fig. \ref{fig:5-2} as follows.
According to Lemma 1, it is difficult to distinguish the curves for ADMM and MRS, when the maximum transmit power at each ST is lower than $15\text{~dBm},$ since the module size constraint is inactive in this regime for $K=10, M=10,$ and $\delta=5.5$.
Instead, the considered ADMM and MRS lead to different triggered module subsets, when $p^{\max}>15\text{~dBm},$ consequently, to different performance of SINR, total transmit power, and EE.
As seen in Fig. \ref{fig:7-3}, the EE achieved by all the schemes increases monotonically with the increase of $p^{\max},$ when $p^{\max}\leq 10\text{~dBm},$ since the total transmit power consumption is negligible with respect to the circuit power consumption and, also, the cochannel interference is small compared to the noise power.
In addition,  the AF relay scheme outperforms the other three algorithms when $p^{\max}\leq 17\text{~dBm},$ since the relay transmit power $P_r^{\max}$ is relevant to the SINR.
Therefore, the increase of SINR (sum rate) dominates the EE in this regime.
For larger values of $p^{\max},$ e.g., $p^{\max}>15\text{~dBm},$ the EE achieved by AF relay decreases and  gradually becomes  inferior to the ADMM, with the increase of $p^{\max}.$
This can be explained as follows.
According to Lemma 1 and Fig. \ref{fig:5-2}, for $K=10, M=10, $ and $\delta=5.5,$ the total static power consumption decreases with increasing $p^{\max},$ since a larger $p^{\max}$ implies a sparser solution.
This also means that the increase of transmit power consumption dominates the EE in this regime and its EE decreases rapidly when $p^{\max}>15\text{~dBm}.$
}

\subsubsection{{Impact of $L$}}
{To describe the impact of $L$ on the system performance, Fig. \ref{fig:8} demonstrates the performance achieved by all the methods using $K=10, M=10$ and $\delta=5.5$ with the maximum transmit power at each ST is $p^{\max}=20\text{~dBm}$,  as a function of the number of reflecting elements at each module, $L.$
In Fig. \ref{fig:8-1}, we can observe that the SINR performance of ADMM and MRS slightly increases with the increasing number of reflecting elements at each module.
In fact, for any given values of $K, M, P^{\max}$, and $\delta,$ it can be seen from Lemma 1 that the cardinality of triggered module subset decreases slightly, while the available reflecting element quantity increases,  with increasing $L.$
Correspondingly, in Fig. \ref{fig:8-3}, the EE achieved by both ADMM and MRS first increases and then decreases with increasing the number of reflecting elements at each module $L.$
In fact, employing more reflecting elements at each module can increase the SINR gain at the expense of more power consumption.
In contrast, from Fig. \ref{fig:8-3}, we exhibit that the EE for AF relay decreases monotonically with the increase $L.$
According to Fig. \ref{fig:7-3}, for $K=10, M=10, p^{\max}=20\text{~dBm},$ and $\delta=5.5,$ the EE of AF relay relies on the total power consumption dominated by the static power consumption and its EE decreases rapidly with increasing of $L.$
In addition, for $K=10, M=10, p^{\max}=20\text{~dBm},$ there is an optimal choice of $L=20,$ which leads to the maximum EE.
}
\begin{figure}[!t]
	\centering
	\begin{minipage}[t]{0.8\linewidth}
		\centering
		\includegraphics[width=1\linewidth]{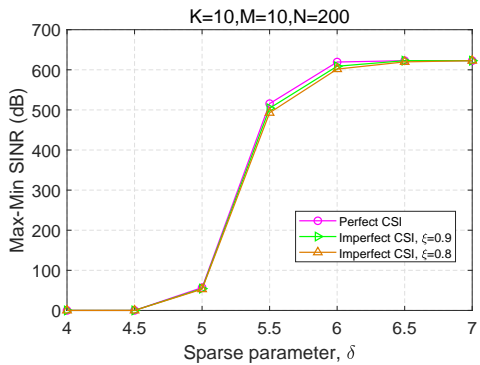}
	\end{minipage}
	\caption{Max-min SINR at different sparse parameter $\delta$, under different levels of reliability of estimate ($K=10, M=10, N=200$).}
	\label{fig:9}
\end{figure}
\subsubsection{{Impact of Imperfect CSI}}
{ Furthermore, Fig. \ref{fig:9} evaluate the impact of imperfect CSI on the performance of the proposed algorithm. Following \cite{Yang2019Low}, the estimated channel vector can be modeled as 
$\hat{\mathbf h}_k=\xi{\mathbf h}_k+\sqrt{1-\xi^2}\vartriangle\hspace{-0.3em}{\mathbf h}_k, \forall k\in{\cal K},$ where $0\leq \xi\leq 1$ represents the level of reliability of the estimate, $\vartriangle\hspace{-0.3em}{\mathbf h}_k\sim{\cal C}{\cal N}(0,\sigma_{\vartriangle\hspace{-0.1em}{\mathbf h}_k}^2{\mathbf I})$ and $\sigma_{\vartriangle\hspace{-0.1em}{\mathbf h}_k}^2=(200/\min_k\{d_{{\mathbf h}_k}\})^2$. The similar model is also used for the estimated channel of ${\mathbf g}_k$, i.e., $\hat{\mathbf g}_k=\xi{\mathbf g}_k+\sqrt{1-\xi^2}\vartriangle\hspace{-0.3em}{\mathbf g}_k.$
Based on the above, in Fig. \ref{fig:9}, we plot the max-min SINR performance of the proposed two-block ADMM algorithm under imperfect CSI for $\delta$ from $4$ to $7$, with $K=10, M=10,$ and $N=200.$
We set the level of reliability of the estimate as $\xi=0.9$ and $ 0.8.$
Fig. \ref{fig:9} shows that the maximum SINR performance loss of the proposed two-block ADMM algorithm with imperfect  CSI ($\xi=0.9$ and $\xi=0.8$) compared to perfect CSI is approximately $7\text{~dB}$ and $10\text{~dB}$, respectively.  
From Fig. \ref{fig:9}, we can observe that the proposed two-block ADMM algorithm is able to cope with imperfect CSI. 
}
\section{Conclusions}

In this paper, we studied the joint problem of active-passive beamforming to IRS-aided P2P communication networks with reflection resource management.
Specifically, the \textit{true} reflection resource management can be realized via triggered modules identification, which builds on the premise of introduction of the modular IRS structure, i.e., the entire IRS consisting of multiple independent and controllable modules.
The objective is to jointly optimize the power allocation and the sparse passive beamforming such that the mimimum SINR is maximized, subject to the power budget at each ST and module size constraint.
To effectively solve the max-min SINR problem, the approximate convex problem was developed for it by using the mixed $\ell_{1,F}\text{-norm}$ in the perspective of group sparse optimization.
The two-block ADMM algorithm was proposed to identify the triggered module subset.
Subsequently, transmit power allocation and the corresponding phase shift for max-min SINR problem without the module size constraint was studied while simultaneously providing STs' power budget.
Finally,  the simulation results demonstrated the convergence and effectiveness of the ADMM algorithm.
In addition, further performance comparison indicated that the introduced modular IRS structure is valuable.

\section*{Appendix A\\
The Proof of Lemma 1}
{
Based on the norm inequality $\left(\sum_{m=1}^M \tilde{\pmb\phi}^m\right)^2\leq M \sum_{m=1}^M(\tilde{\pmb\phi}^m)^2,$ with $\tilde{\pmb\phi}^m$ is the $m\text{th}$ entry of vector $\tilde{\pmb\phi},$ we have
\begin{equation}
\begin{aligned}
\left(\sum_{m=1}^M ||\bar{\pmb\Phi}^m||_F\right)^2&\leq M \sum_{m=1}^M ||\bar{\pmb\Phi}^m||_F^2
=M||\bar{\pmb\Phi}||_F^2\\
&\leq MK\max_{k}\{p_k^{\max}\} ||\pmb\phi||_2^2\\
&\leq MKN\max_{k}\{p_k^{\max}\}.
\end{aligned}
\end{equation}
Therefore, the module size constraint becomes inactive when $\delta(\delta+0.01)\geq \sqrt{MKN\max_{k}\{p_k^{\max}\}},$ i.e., all the $M$ modules tend to be triggered for cooperative communication.
}

\section*{Appendix B\\
The Proof of Theorem 1}

{\em{1) Optimization for problem ${ P}_{\bar{\pmb\phi}_k}:$}}
We rewrite ${\mathbf W}$, ${\pmb\Lambda},$ ${\pmb\Psi}$, and ${\mathbf F}$ in forms of
${\mathbf W}=[{\mathbf w}^1, {\mathbf w}^2, \ldots, {\mathbf w}^K],$ ${\pmb\Lambda}=[{\pmb\lambda}^1, {\pmb\lambda}^2, \ldots, {\pmb\lambda}^K],$
${\pmb\Psi}=\left[{\pmb\psi}^{1}, {\pmb\psi}^{2}, \ldots, {\pmb\psi}^{K+1} \right],$ and ${\mathbf F}=\left[{\mathbf f}^{1}, {\mathbf f}^{2}, \ldots, {\mathbf f}^{K} \right]$, respectively.
And ${\mathbf w}^k\in {\mathbb C}^{N\times 1}$ and ${\pmb\lambda}^k\in {\mathbb C}^{N\times 1}$ represent the $k\text{-th}$ column of matrices ${\mathbf W}$ and ${\pmb\Lambda}$, respectively.
${\pmb\psi}^{k}\in{\mathbb C}^{K\times 1}$ and ${\mathbf f}^{k}\in {\mathbb C}^{K\times 1} $
are the $k\text{-th}$ column of matrices ${\pmb\Psi}$ and ${\mathbf F}$, respectively.

Introducing auxiliary matrix $\tilde{\mathbf h}^k=[\bar{\mathbf h}_{k,1}, \bar{\mathbf h}_{k,2}, \ldots, \bar{\mathbf h}_{k,K}]\in{\cal C}^{N\times K},$ then, $\bar{\pmb\Phi}$ is separable in ${ P}_{\bar{\pmb\phi}_k}$. The optimization problem of $\bar{\pmb\phi}_k, \forall k\in{\cal K}$ is
\begin{equation}\label{s:22}
\begin{aligned}
\min_{\bar{\pmb\phi}_k}& ~\text{Re}\{\text{Tr}[{\pmb\lambda}^{k\dag}({\mathbf w}^k(t)-\bar{\pmb\phi}_k) ]\}+\frac{c}{2}||{\mathbf w}^k(t)-\bar{\pmb\phi}_k ||_2^2\\
&+\text{Re}\{\text{Tr}[{\pmb\psi}^{k \dag}({\mathbf f}^{k}(t)-\tilde{\mathbf h}^{k \dag}\bar{\pmb\phi}_k)]\}\\
&+\frac{c}{2}\left\|{\mathbf f}^{k}(t)-\bar{\mathbf h}^{k \dag}\bar{\pmb\phi}_k \right\|_2^2\\
\text{s. t. ~}& \bar{\pmb\phi}_k^{\dag}{\mathbf e}_n{\mathbf e}_n^{\dag}\bar{\pmb\phi}_k\leq p_k^{\max}, \forall n\in{\cal N}; k\in{\cal K},
\end{aligned}
\end{equation}
which can be easily solved by exploiting the first-order optimality condition. Specifically, we have
\begin{equation}\label{s:23}
\begin{aligned}
&-{\pmb\lambda}^k+c(\bar{\pmb\phi}_k-{\mathbf w}^k(t))-\tilde{\mathbf h}^k{\pmb\psi}^{k}(t)
+c(\tilde{\mathbf h}^k\tilde{\mathbf h}^{k \dag}\bar{\pmb\phi}_k-\tilde{\mathbf h}^k{\mathbf f}^{k}(t)  )\\
&+2(\sum_{n=1}^N\mu_n^k{\mathbf e}_n{\mathbf e}_n^{\dag})\bar{\pmb\phi}_k=0,
\end{aligned}
\end{equation}
and consequently, the optimal solution of ${\pmb\phi}$ is given by
\begin{equation}\label{s:24}
\begin{aligned}
\bar{\pmb\phi}_k(t+1)=&( c{\mathbf I}_{N\times N}+c\tilde{\mathbf h}^k\tilde{\mathbf h}^{k\dag}+2\sum_{n=1}^N \mu_n^k{\mathbf e}_n{\mathbf e}_n^{\dag})^{-1}\\
&\times({\pmb\lambda}^k(t)+c{\mathbf w}^k(t)+\tilde{\mathbf h}^k{\pmb\psi}^{k}(t)
+c\tilde{\mathbf h}^k{\mathbf f}^{k}(t)),\\
\end{aligned}
\end{equation}
where $\mu_n^k$ is the Lagrangian multipliers of $\bar{\pmb\phi}_k^{\dag}{\mathbf e}_n{\mathbf e}_n^{\dag}\bar{\pmb\phi}_k \leq p_k^{\max}, \forall n=1, 2, \ldots, N,$ and should be properly chosen to satisfy the KKT condition \cite{Boyd2009Convex}. And the dual variable $\mu_n^k$ is optimally determined by
\begin{equation}
\mu_n^k=\max\{ p_k^{\max}-\bar{\pmb\phi}_k^{\dag}(t+1){\mathbf e}_n{\mathbf e}_n^{\dag}\bar{\pmb\phi}_k(t+1), 0\}.
\end{equation}

{\em{2) Optimization for Problem ${P}_{\mathbf W}:$}}
The problem of ${\mathbf W}$ is an unconstrained group \textit{Lasso} problem, i.e.,
\begin{equation}\label{s:25}
\begin{aligned}
\text{P}({\mathbf W}):  \min_{\mathbf W} \sum_{m=1}^{M}& \alpha||{\mathbf W}^m||_2
+\text{Re}\{\text{Tr}[{\pmb\Lambda}^{\dag}({\mathbf W}-\bar{\pmb\Phi}(t+1))] \}\\
&+\frac{c}{2}|| {\mathbf W}-\bar{\pmb\Phi}(t+1) ||_2^2.
\end{aligned}
\end{equation}
Let ${\pmb\Lambda}^m\in {\mathbb C}^{L\times K},$ ${\mathbf W}^m\in {\mathbb C}^{L\times K},$ and $\bar{\pmb\Phi}^m\in {\mathbb C}^{L\times K}$
be the $m\text{-th}$ row block of matrices ${\pmb\Lambda}, {\mathbf W},$ and $\bar{\pmb\Phi}$, respectively,  $\forall m=1, 2, \ldots, M.$
Then, ${\cal P}_{\mathbf W}$ can be divided into $M$ independent problems of ${\mathbf W}^m$ for $m=1, 2, \ldots, M$
\begin{equation}\label{s:26}
\begin{aligned}
\text{P}({\mathbf W}^m):&\min_{{\mathbf W}^m}\alpha||{\mathbf W}^m||_2+
\frac{c}{2}|| {\mathbf W}^m-\bar{\pmb\Phi}^m(t+1)||_2^2\\
&+\text{Re}\left\{\text{Tr}\left[{\pmb\Lambda}^{m}(t)\left({\mathbf W}^{m}-\bar{\pmb\Phi}^m(t+1)\right)\right]  \right\}.
\end{aligned}
\end{equation}
The first-order optimality condition for the optimal solution ${\mathbf W}^m(t+1)$ we have
\begin{equation}\label{s:27}
\begin{aligned}
\alpha\partial||{\mathbf W}^m(t+1)||_2=\underbrace{c\bar{\pmb\Phi}^m(t+1)
-{\pmb\Lambda}^m(t)}_{{\pmb\Xi}^m(t)}-c{\mathbf W}^m(t+1),
\end{aligned}
\end{equation}
where $\partial|| {\mathbf W}^m(t+1) ||_2$ is the subgradient of $|| {\mathbf W}^m(t+1)||_2$ defined as
\begin{equation}\label{s:28}
\partial|| {\mathbf W}^m(t+1) ||_2=\frac{{\mathbf W}^m(t+1)}{||{\mathbf W}^m(t+1)||_2}.
\end{equation}
Inserting (\ref{s:28}) into (\ref{s:27}), we can easily obtain
\begin{equation}\label{s:29}
\begin{aligned}
{\mathbf W}^m(t+1)=\left\{\begin{array}{ll}
&{\mathbf 0}, \text{~if~} ||{\pmb\Xi}(t)||_2\leq \alpha\\
&\frac{(||{\pmb\Xi}^m(t)||_2-\alpha){\pmb\Xi}^m(t)}
{c||{\pmb\Xi}^m(t)||_2}, \text{~otherwise}.
\end{array}\right.
\end{aligned}
\end{equation}

\section*{Appendix C\\
The Proof of Theorem 2}
The first-order optimality conditions for ${\mathbf f}_k(t+1)$ are listed as follows:
\begin{align}
&\psi_{k,k}(t)+c\left[f_{k,k}(t+1)-b_{k,k}(t+1) \right]-\frac{\varepsilon_k}{\sqrt{\gamma^{-1}}}=0 \label{s:32-1}\\
& -{\pmb\psi}_{-k,k}(t)-c\left( {\mathbf f}_{-k,k}(t+1)-{\mathbf b}_{-k,k}(t+1)\right)\nonumber\\
&~~~~~~~~~~~~~~~~~~~~~~~~~~~~~~~~~~~=\varepsilon_k\partial|| {\mathbf f}_{-k,k}(t+1) ||_2\label{s:32-2}\\
&\varepsilon_k\left( \frac{f_{k,k}(t+1)}{\sqrt{\gamma^{-1}}}-||{\mathbf f}_{-k,k}(t+1)||_2 \right)=0\label{s:32-3}\\
&\varepsilon_k\geq 0 \label{s:32-4}\\
&\frac{f_{k,k}(t+1)}{\sqrt{\gamma^{-1}}}-||{\mathbf f}_{-k,k}(t+1)||_2 \geq 0,\label{s:32-5}
\end{align}
where $\varepsilon_k$ is the Lagrangian multiplier for $\sqrt{\gamma^{-1}} f_{k,k}\geq ||{\mathbf f}_{-k,k} ||_2$.

Assume that ${\mathbf f}_{k,-k}(t+1)\neq {\mathbf 0}$ first, from (\ref{s:32-1}) and (\ref{s:32-2}), we can easily get
\begin{equation}\label{s:33}
\left\{\begin{array}{ll}
&f_{k,k}(t+1)=\frac{cb_{k,k}(t+1)-\psi_{k,k}(t)+\sqrt{\gamma^{-1}}\varepsilon_k}{c}\\
&{\mathbf f}_{-k,k}(t+1)=\frac{c{\mathbf b}_{-k,k}(t+1)-{\pmb\psi}_{-k,k}(t)}
{c+\varepsilon_k\rho_k},
\end{array}\right.
\end{equation}
where {$\rho_k=(||{\mathbf f}_{-k,k}(t+1)||_F)^{-1}$}, $\varepsilon_k$ should be properly chosen such that KKT complementary condition should be satisfied.
If $\left( \frac{f_{k,k}(t+1)}{\sqrt{\gamma}}-||{\mathbf f}_{-k,k}(t+1)||_2 \right)\big|_{\varepsilon_k=0}\geq 0$ or equivalently
$\sqrt{\gamma^{-1}}\left(c{b}_{k,k}(t+1)-\psi_{k,k}(t)\right)\geq ||c{\mathbf b}_{-k,k}(t+1)-{\pmb\psi}_{-k,k}(t)||_2,$
we have $\varepsilon_k=0.$
Otherwise, we have $\sqrt{\gamma^{-1}}f_{k,k}(t+1)=||{\mathbf f}_{-k,k}(t+1)||_2$ for some $\varepsilon_k>0.$
In the case of $||c{\mathbf b}_{-k,k}(t+1)-{\pmb\psi}_{-k,k}(t)||_2>\sqrt{\gamma^{-1}}(cb_{k,k}(t+1)-\psi_{k,k}(t)).$
Combining $\rho_k||{\mathbf f}_{-k,k}||_2=1$ and $\sqrt{\gamma^{-1}}f_{k,k}(t+1)=||{\mathbf f}_{-k,k}(t+1)||_2$, we obtain
\begin{equation}\label{sss:1}
\begin{aligned}
\varepsilon_k=&\frac{1}{1+\gamma}\left[\gamma||c{\mathbf b}_{-k,k}(t+1)-{\pmb\psi}_{-k,k}(t) ||_2\right.\\ &\left. -\sqrt{\gamma}(cb_{k,k}(t+1)-\psi_{k,k}(t))\right]\\
\rho_k=&\frac{1+\gamma}{c^{-1}}\left[ c||{\mathbf b}_{-k,k}(t+1)-{\pmb\psi}_{-k,k}(t)||_2\right.\\
&\left.+\sqrt{\gamma}(cb_{k,k}(t+1)-\psi_{k,k}(t))\right].
\end{aligned}
\end{equation}


\begin{thebibliography}{99}
	
	\bibitem{whitepaper}
	``Mobile-edge computing--introductory technical white paper,''
	{\em ETSIM}, 2014. [Online]. Available: 
	\url{https://www.scirp.org/reference/ReferencesPapers.aspx?ReferenceID=2453997}
	
	
	\bibitem{Grassi2017Uplink}
	A.~Grassi, G.~Piro, G.~Bacci, and G.~Boggia, ``Uplink resource management in
	{5G}: when a distributed and energy-efficient solution meets power and {QoS}
	constraints,'' {\em IEEE Trans.  Veh. Tech.}, vol.~66, no.~6,
	pp.~5176--5189, Jun. 2017.
	
	\bibitem{Chen2014The}
	S.~Chen and J.~Zhao, ``The requirements, challenges, and technologies for {5G}
	of terrestrial mobile telecommunication,'' {\em IEEE Commun. Magz.},
	vol.~52, no.~5, pp.~36--43, May 2014.
	
	\bibitem{Boccardi2014Five}
	F.~Boccardi, R.~W. Heath, A.~Lozano, T.~L. Marzetta, and P.~Popovski, ``Five
	disruptive technology directions for {5G},'' {\em IEEE Commun.
		Mag.}, vol.~52, no.~2, pp.~74--80, Feb., 2014.
	
	\bibitem{NGMN2015}
	``Ngmn alliance {5G} white paper (2015).'' [Online]. Available:
	\url{https://s3.amazonaws.com/academia.edu.documents/56609953/NGMN_5G_White_Paper_V1_0.pdf?response-content-disposition=inline}
	
	
	\bibitem{Zenzo2019Wireless}
	E.~Basar, M.~D. Renzo, J.~D. Rosny, M.~Debbah, M.~S. Alouini, and R.~Zhang,
	``Wireless communications through reconfigurable intelligent surfaces,'' {\em
		IEEE Access}, vol.~7, pp.~116753--116773, Aug., 2019.
	
	\bibitem{Wu2019Towards}
	Q.~Wu and R.~Zhang, ``Towards smart and reconfigurable environment: intellignet
	reflecting surfaces aided wireless network,'' {\em IEEE Commun. Mag.},
	vol.~58, no.~1, pp.~106--112, Nov., 2019.
	2019.
	
	\bibitem{Tan2018Enabling}
	X.~Tan, Z.~Sun, D.~Koutsonikolas, and J.~M. Jornet, ``Enabling indoor mobile
	millimeter-wave networks based on smart reflect-arrays,'' in {\em Proc. 
		INFOCOM}, Honolulu, HI, USA,  Apr., 2018, pp.~1--9.
	
	\bibitem{Liu2019Intelligent}
	F.~Liu, O.~Tsilipakos, A.~Pitilakis, A.~C. Tasolamprou, M.~S. Mirmoosa, N.~V.
	Kantartzis, and et. al., ``Intelligent metasurfaces with continuously tunable
	local surface impedance for multiple reconfigurable functions,'' {\em
		Physical Review Applied}, vol.~11, NO.~4, pp.~044024, Apr., 2019.
	
	\bibitem{Li2017Electromagnetic}
	L.~Li, C.~T. Jun, W.~Ji, S.~Liu, J.~Ding, X.~Wan, L.~Y. Bo, M.~Jiang, C.~Qiu,
	and S.~Zhang, ``Electromagnetic reprogrammable coding-metasurface
	holograms,'' {\em Nature Commun.}, vol.~8, NO.~1, pp.~1--7, Aug., 2017.
	
	\bibitem{Di2019Hybrid}
	B.~Di, H.~Zhang, L.~Song, Y.~Li, Z.~Han, and H.~Vincent~Poor, ``Hybrid
	beamforming for reconfigurable intelligent surface based multi-user
	communications: Achievable rates with limited discrete phase shifts,''
	{\em IEEE J. Sel. Areas Commun.}, to be published. [Online]. Available:
	\url{https://ieeexplore.ieee.org/stamp/stamp.jsp?tp=&arnumber=9110889}
	
	
	\bibitem{Gao2020Reconfigurable}
	Y.~Gao, C.~Yong, Z.~H. Xiong, D.~Niyato, and Y.~Xiao, ``Reflection resource
	management for intelligent reflecting surface aided wireless networks,'' in {\em Proc. ICC,} to be published.   [Online]. Available:  \url{https://arxiv.org/abs/2002.00331}
	
	\bibitem{Wu2018Intelligent}
	Q.~Wu and R.~Zhang, ``Intelligent reflecting surface enhanced wireless network:
	Joint active and passive beamforming design,'' in {\em  Proc. IEEE GLOBECOM}, Abu Dhabi, United Arab Emirates, pp.~1--6, Dec., 2018.
	
	\bibitem{Guo2019Weighted}
	H.~Guo, Y.~C. Liang, J.~Chen, and E.~G. Larsson, ``Weighted sum-rate
	optimization for intelligent reflecting surface enhanced wireless networks,''
	2019. [Online]. Available: \url{https://arxiv.org/abs/1905.07920}
	
	\bibitem{Wu2019Beamforming}
	Q.~Wu and R.~Zhang, ``Beamforming optimization for wireless network aided by
	intelligent reflecting surface with discrete phase shifts,''
	{\em IEEE Trans. Wireless Commun.}, vol.~68,  no.~3, pp.~1--30, Dec., 2019.
	
	\bibitem{Bjornson2019Intelligent}
	E.~Bjornson, O.~Ozdogan, and E.~G. Larsson, ``Intelligent reflecting surface
	vs. decode-and-forward: How large surfaces are needed to beat relaying?,''
	{\em IEEE Wireless Commun. Lett.}, vol.~9  no.~2
	pp.~244--248, Feb., 2020.
	
	\bibitem{Han2018Large}
	Y.~Han, W.~Tang, S.~Jin, C.~K. Wen, and X.~Ma, ``Large intelligent
	surface-assisted wireless communication exploiting statistical {CSI},'' {\em
		IEEE Trans.  Veh. Tech.}, vol.~68, no.~8, pp.~8238--8242,  Aug., 2018.
	
	\bibitem{Wu2019Intelligentjournal}
	Q.~Wu and R.~Zhang, ``Intelligent reflecting surface enhanced wireless network
	via joint active and passive beamforming design,'' {\em IEEE Trans.
		Wireless Commun.}, vol.~18,  no.~11, pp.~5394--5409,  Nov.,  2019.
	
	\bibitem{Huang2019Reconfigurable}
	C.~Huang, A.~Zappone, G.~C. Alexandropoulos, M.~Debbah, and C.~Yuen,
	``Reconfigurable intelligent surfaces for energy efficiency in wireless
	communication,'' {\em IEEE Trans.  Wireless Commun.}, vol.~18, no.~8,
	pp.~4157--4170, Aug., 2019.
	
	\bibitem{Zenzo2019Reconfigurable}
	K.~Ntontin, M.~D. Renzo, J.~Song, F.~Lazarakis, J.~D. Rosny, D.-T. Phan-Huy,
	O.~Simeone, R.~Zhang, M.~Debbah, G.~Lerosey, M.~Fink, S.~Tretyakov, and
	S.~Shamai, ``Reconfigurable intelligent surfaces vs. relaying: Differences,
	similarities, and performance comparison,'' 2019. [Online]. Available:
	\url{https://arxiv.org/abs/1908.08747}
	
	\bibitem{huang2020holographic}
	C.~Huang, S.~Hu, G.~C. Alexandropoulos, A.~Zappone, C.~Yuen, R.~Zhang,
	M.~Di~Renzo, and M.~Debbah, ``Holographic mimo surfaces for 6{G} wireless
	networks: Opportunities, challenges, and trends,'' {\em IEEE Wireless
		Commun.}, to be published. [Online]. Available: \url{https://ieeexplore.ieee.org/stamp/stamp.jsp?tp=&arnumber=9136592}
	
	\bibitem{huang2020reconfigurable}
	C.~Huang, R.~Mo, and C.~Yuen, ``Reconfigurable intelligent surface assisted multiuser MISO systems exploiting deep reinforcement learning,'' 2020. [Online]. Available: \url{https://arxiv.org/pdf/2002.10072.pdf}
	
	\bibitem{Mehanna2013Joint}
	O.~Mehanna, N.~D. Sidiropoulos, and G.~B. Giannakis, ``Joint multicast
	beamforming and antenna selection,'' {\em IEEE Trans. Signal
		Process.}, vol.~61, no.~10,  pp.~2660--2674, May 2013.
	
	\bibitem{Zenzo2019Smart}
	M.~D. Renzo, M.~Debbah, D.~T. Phan-Huy, and A.~Zapppone, ``Smart radio
	environments empowered by reconfigurable {AI} meta-surfaces: an idea whose
	time has come,'' {\em Eurasip J. Wireless Commun.
		Net.}, no.~1, pp.~1--20, May 2019.
	
	\bibitem{Feng2013Device}
	D.~Feng, Y.~Lu, Y.~Yuan-Wu, G.~Y. Li, G.~Feng, and S.~Li, ``Device-to-device
	communications underlaying cellular networks,'' {\em IEEE Trans.
		Commun.}, vol.~61, no.~8, pp.~3541--3551, Aug., 2013.
	
	\bibitem{Crouzeix1985An}
	J.~P. Crouzeix, J.~A. Ferland, and S.~Schaible, ``An algorithm for generalized
	fractional programs,'' {\em J. Optimization Theory and Application},
	vol.~47, no.~1,  pp.~35--49, Sep., 1985.
	
	\bibitem{Borde1987Convergence}
	J.~Borde and J.~P. Crouzeix, ``Convergence of a dinkelbach-type algorithm in
	generalized fractional programming,'' {\em Zeitschrift f{\"u}r Operations
		Research}, vol.~31, no.~1, pp.~A 31--A 54, Jan., 1987.
	
	\bibitem{Benadada1988Partial}
	Y.~Benadada and J.~A. Fedand, ``Partial linearization for generalized
	fractional programming,'' {\em ZOR--Zeitschrift f{\"u}r Operations Research},
	vol.~32, no.~2, pp.~101--106, Mar., 1988.
	
	\bibitem{Fukushima1992Application}
	M.~Fukushima, ``Application of the alternating direction method of multipliers
	to separable convex programming problems,'' {\em Computing Optimization
		Application}, vol.~1, no.~1, pp.~93--111, Oct., 1992.
	
	\bibitem{Candes2008Enhancing}
	E.~Candes, M.~Wakin, and S.~Boyd, ``Enhancing sparsity by reweighted $\ell_1$
	minimization,'' {\em Journal Fourier Analysis Application}, vol.~14, no.~5,
	pp.~877--905, Dec.,  2008.
	
	\bibitem{Yuan2006Model}
	M.~Yuan and Y.~Lin, ``Model selection and estimation in regression with grouped
	variables,'' {\em J. Royal Statistical Society}, vol.~68, no.~1, pp.~49--67,
	Dec., 2006.
	
	\bibitem{Lin2016Joint}
	J.~Lin, Q.~Li, C.~Jiang, and H.~Shao, ``Joint multirelay selection, power
	allocation, and beamformer design for multiuser decode-and-forward relay
	networks,'' {\em IEEE Trans.  Veh. Tech.}, vol.~65, no.~7,
	pp.~5073--5087, Jul., 2016.
	
	\bibitem{Taylor1979Deconvolution}
	H.~L. Taylor, S.~C. Banks, and J.~F. McCoy, ``Deconvolution with the $\ell_1$
	norm,'' {\em Geophysics}, vol.~44, no.~1, pp.~39--52, 1979.
	
	\bibitem{Santosa1986Linear}
	F.~Santosa and W.~W. Symes, ``Linear inversion of band-limited reflection
	seismograms,'' {\em SIAM Journal on Scientific and Statistical Computing},
	vol.~7, no.~4, pp.~1307--1330, 1986.
	
	\bibitem{Grant2011The}
	M.~C. Grant and S.~P. Boyd, ``The cvx users' guide,'' 2014. [Online]. Available:
	\url{http://cvxr.com/cvx/doc/}
	
	\bibitem{Boyd2009Convex}
	S.~Boyd and L.~Vandenberghe, {\em Convex Optimization}.
	\newblock Cambridge, U.K.: Cambridge University Press, 2009.
	
	
	
	
	\bibitem{Hestenes1969Multiplier}
	M.~R. Hestenes, ``Multiplier and gradient methods,'' {\em Journal Optimization
		Theory and Application}, vol.~4, no.~5,  pp.~303--320, Nov., 1969.
	
	\bibitem{Ramamonjison2015Energy}
	R.~Ramamonjison and V.~K. Bhargava, ``Energy efficiency maximization framework
	in cognitive downlink two-tier networks,'' {\em IEEE Trans. Wireless
		Commun.}, vol.~14, no.~3, pp.~1468--1479, Mar.,  2014.
	
	\bibitem{Csiszar1984Information}
	I.~Csiszar and G.~Tusnady, ``Information geometry and alternating minimization
	procedures,'' {\em Statistic Decisions}, vol.~1, pp.~205--237, Dec., 1984.
	
	\bibitem{Zheng2020Intelligent}
	B.~Zheng, Q.~Wu, and R.~Zhang, ``Intelligent reflecting surface-assisted
	multiple access with user pairing: NoMA or OMA?,'' {\em IEEE Commun.
		Lett.}, vol.~24, no.~4, pp.~753--757, Apr., 2020.
	
	\bibitem{Yang2019Low}
	Z.~Yang and M.~Dong, ``Low-complexity coordinated relay beamforming design for
	multi-cluster relay interference networks,'' {\em IEEE Trans.
		Wireless Commun.}, vol.~18, no.~4, pp.~2215--2228, Apr., 2019.
	
\end{thebibliography}

\end{document}